\documentclass[11pt]{article}

\usepackage[margin=1in]{geometry}
\usepackage[T1]{fontenc}
\usepackage{lmodern}
\usepackage{amsmath,amssymb,amsthm,mathtools,bm}
\usepackage{enumitem}
\usepackage{microtype}
\usepackage{booktabs,tabularx,array}
\usepackage{aliascnt}
\usepackage[colorlinks=true,linkcolor=blue,citecolor=blue,urlcolor=blue]{hyperref}
\usepackage[nameinlink,capitalise,noabbrev]{cleveref}
\usepackage{tikz}

\newtheorem{theorem}{Theorem}[section]

\newaliascnt{lemma}{theorem}
\newtheorem{lemma}[lemma]{Lemma}
\aliascntresetthe{lemma}

\newaliascnt{proposition}{theorem}
\newtheorem{proposition}[proposition]{Proposition}
\aliascntresetthe{proposition}

\newaliascnt{assumption}{theorem}
\newtheorem{assumption}[assumption]{Assumption}
\aliascntresetthe{assumption}

\newaliascnt{remark}{theorem}

\aliascntresetthe{remark}

\newaliascnt{definition}{theorem}

\aliascntresetthe{definition}

\newaliascnt{corollary}{theorem}

\aliascntresetthe{corollary}

\theoremstyle{definition}
\newaliascnt{algorithm}{theorem}

\aliascntresetthe{algorithm}
\theoremstyle{plain}

\crefname{theorem}{Theorem}{Theorems}
\Crefname{theorem}{Theorem}{Theorems}
\crefname{lemma}{Lemma}{Lemmas}
\Crefname{lemma}{Lemma}{Lemmas}
\crefname{proposition}{Proposition}{Propositions}
\Crefname{proposition}{Proposition}{Propositions}
\crefname{assumption}{Assumption}{Assumptions}
\Crefname{assumption}{Assumption}{Assumptions}
\crefname{remark}{Remark}{Remarks}
\Crefname{remark}{Remark}{Remarks}
\crefname{definition}{Definition}{Definitions}
\Crefname{definition}{Definition}{Definitions}
\crefname{corollary}{Corollary}{Corollaries}
\Crefname{corollary}{Corollary}{Corollaries}
\crefname{algorithm}{Algorithm}{Algorithms}
\Crefname{algorithm}{Algorithm}{Algorithms}

\newcommand{\R}{\mathbb R}
\newcommand{\C}{\mathbb C}
\newcommand{\dd}{\,\mathrm d}
\newcommand{\ii}{ i}
\newcommand{\eps}{\epsilon}
\newcommand{\bigO}{\mathcal O}
\newcommand{\widetildeO}{\widetilde{\mathcal O}}
\newcommand{\polylog}{\operatorname{polylog}}

\newcommand{\diag}{\operatorname{diag}}

\newcommand{\Lbox}{L_{\rm box}}

\newcommand{\cmin}{c_{\min}}
\newcommand{\cmax}{c_{\max}}
\newcommand{\norm}[1]{\left\lVert #1\right\rVert}

\newcommand{\by}{\bm y}
\newcommand{\bq}{\bm q}
\newcommand{\be}{\bm e}

\newcommand{\bg}{\bm g}

\newcommand{\calM}{ M}
\newcommand{\calC}{ C}

\title{Exponential Reduction of Mesh Dependence in Quantum Estimation
of Parabolic PDE Observables}
\author{Xiantao Li\\ Pennsylvania State University\\ \texttt{Xiantao.Li@psu.edu}}
\date{July 2026}

\begin{document}
\maketitle

\begin{abstract}
Can a quantum PDE algorithm avoid the polynomial cost of resolving a fine
spatial mesh?  For standard fixed-order discretizations, direct classical
methods require work polynomial in \(h^{-1}\), or equivalently in the number
of spatial degrees of freedom \(N_h=\Theta(h^{-d})\).  Even with QSVT applied to a discrete-gradient block encoding,
direct implementation of the parabolic semigroup has coherent
complexity $\widetilde{\mathcal O}(\sqrt{T}/h)$, while
gradient-dependent observables introduce additional mesh-dependent
readout costs. 
Decay of the solution norm will further suppress the postselection
probability for preparing a normalized final state.

We develop a multilevel quantum algorithm that estimates linear and
quadratic observables directly, and leverages the fine--coarse cancellation
inside the circuit to alleviate the measurement cost.  A shifted Ritz--Schur
factorization identifies the corrected resolvent and provides sufficient
conditions for preserving its $\bigO(h_\ell^2)$ normalization. The time evolution is formulated as a
contour-based LCU reconstruction, where each target-time correction from a
coherently selected family of corrected resolvents. We verify the
required access for nested sine--Galerkin spaces and for uniform dyadic
$P_1$ elements in one dimension with diffusion that is piecewise constant
on an aligned coarsest mesh. The latter construction uses
energy-orthogonal midpoint details, an explicit discrete Green matrix,
and a stable resolvent identity.

For readout normalization order $0\le\chi\le2$ and $T\ge1$,
optimized amplitude estimation removes polynomial dependence on the
finest mesh size. Under the stated preparation and selected-access
assumptions, finest-level discrete linear and quadratic observables can
be estimated to additive error $\eps$ with complexity
$\widetildeO(1+(T\eps)^{-1})$ and only polylogarithmic dependence on
$h^{-1}$. 
\end{abstract}

\section{Introduction}
\label{sec:intro}

Partial differential equations  (PDEs) are natural candidate applications of
quantum computation because they are central to many engineering models and
their discretizations produce large, structured algebraic systems. Quantum
linear-system, block-encoding, and matrix-function algorithms are therefore
directly applicable, often with polylogarithmic dependence on the algebraic
dimension under suitable input and operator-access models
\cite{HarrowHassidimLloyd2009,ChildsKothariSomma2017,
CostaEtAl2022Optimal,GilyenSuLowWiebe2019}. This observation has led to a substantial literature.  General quantum
algorithms for systems of linear ordinary differential equations provide
basic tools for a semi-discrete approximation of PDEs, see 
\cite{BerryChildsOstranderWang2017,Krovi2023,
FangLinTong2023} and references therein.  Quantum algorithms have also been developed
for Poisson, elliptic, and finite-element boundary-value problems
\cite{CaoEtAl2013Poisson,CladerJacobsSprouse2013,
MontanaroPallister2016,ChildsLiuOstrander2021,
DeimlPeterseim2025}.  Complementary approaches target time-dependent PDEs
directly, including the wave and heat equations, general linear PDEs through
Schr\"odingerization type mapping, and explicit Fourier-space circuit constructions
\cite{costa2019quantum,LindenMontanaroShao2022,liang2026structure,
JinLiuYu2024Schrodingerization,
LubaschKikuchiWrightMcKeever2025FourierPDE,hu2024quantum,alipanah2025quantum}.

 For PDE computations,  however, polylogarithmic dependence on the algebraic
dimension is not by itself an end-to-end complexity result. If a spatial
discretization has mesh size $\bigO(h)$, which needs to be sufficiently small to approximate
differential operators with desired precision, then the operator normalization,
condition number, discretization accuracy, and observable normalization can all
depend on the same parameter $h$, thus leading to a polynomial dependence of the overall complexity on $h^{-1}$.  In particular, mesh-dependent normalization and measurement costs can dominate this final
inference stage even when the coherent linear-algebra primitive is
efficiently preconditioned
\cite{MontanaroPallister2016,Li2026EndToEndElliptic,DeimlPeterseim2025}. The
appropriate computational target is therefore not necessarily a normalized
solution state, but a scalar functional of the unnormalized solution
operator.

Parabolic equations make this distinction especially clear. Consider the
model problem \cite{evans2022partial}
\begin{equation}
  \partial_t u+\mathcal L u=0,
  \,\;
  \mathcal L u=-\nabla\!\cdot\!\bigl(a(x)\nabla u\bigr),
  \label{eq:intro-parabolic-model}
\end{equation}
on a bounded polyhedral domain with homogeneous Dirichlet boundary
conditions. A standard mass-reduced $P_1/Q_1$ finite element discretization
with mesh width $h$ gives a semi-discrete approximation as an
$N_h$-dimensional linear ODE system \cite{Thomee2006},
\[
  \dot{\by}_h=-L_h\by_h,
  \,\;
  L_h=L_h^\dagger\succ 0,
  \,\;
  \norm{L_h}=\bigO(h^{-2}),
\]
and, with appropriate spatial discretization, the discrete generator
also admits a first-order Gram factorization
\[
  L_h=A_h^\dagger A_h,
  \,\;
  \norm{A_h}=\bigO(h^{-1}).
\]
Given a block encoding of $A_h/\alpha_A$, with
$\alpha_A=\Theta(h^{-1})$, QSVT applies an even polynomial
approximation to $x\mapsto e^{-T\alpha_A^2x^2}$ on $[-1,1]$.
The required degree is
$\widetilde{\bigO}(1+\alpha_A\sqrt T)
=\widetilde{\bigO}(1+\sqrt T/h)$.
Thus first-order access yields a square-root dependence on the
second-order scale $T\alpha_A^2$, but the coherent cost still depends
polynomially on $h^{-1}$.

A finite element approximation also inherits the Dirichlet gap
$L_h\succeq\lambda_*I$. Writing $E_h(T):=e^{-TL_h}$, if one were to prepare
the solution $\bm y_h(T)=E_h(T)\bm y_h(0)$ by applying $E_h(T)$ to a
normalized initial state, postselection of the signal ancilla would succeed
with probability
\begin{equation}\label{pT}
    p_T
  =
  \norm{E_h(T)\lvert\by_h(0)\rangle}^2
  \le e^{-2\lambda_*T}.
\end{equation}
Preparing the normalized final heat state can therefore incur an
exponentially large amplification overhead. This is not a numerical
instability of the parabolic equation; it is the consequence of demanding a
normalized encoding of a vector that is physically decaying.

Most PDE calculations  seek
a scalar quantity such as a regional average, heat flux, local mass, or
dissipated energy. After discretization, representative linear and
quadratic quantities of interest take the forms
\begin{equation}\label{readouts}
     \bg_h^\dagger E_h(T)\by_h(0),
  \,\;
  \by_h(0)^\dagger
  E_h(T)^\dagger M_hE_h(T)\by_h(0),
\end{equation}
and both can be estimated directly as matrix elements of the semigroup,
without first preparing the normalized state $\lvert\by_h(T)\rangle$. This
observable-driven formulation avoids the decay-induced final-state
normalization penalty and is the organizing principle of this paper.

Linden, Montanaro, and Shao (LMS) emphasized the importance of this kind of
end-to-end accounting for the heat equation
\cite{LindenMontanaroShao2022}. Their work shows that spatial
discretization, the normalization of a history-state construction,
conversion from discrete amplitudes to physical integrals, and statistical
inference can together dominate the cost of a nominally efficient quantum
linear-system algorithm. They also exhibit a substantially faster method
based on probabilistic heat-kernel access and amplitude estimation in the
constant-coefficient periodic setting.

In this paper we pursue a different direction, targeting more general
settings. We discretize the parabolic PDE \eqref{eq:intro-parabolic-model}
with Galerkin methods (finite element or spectral) and work directly with semigroup matrix elements,
aiming at linear and quadratic observables from the outset. We further
consider observables that depend on the gradient of the solution, for
which direct estimation incurs additional $h$-dependence beyond that of
the underlying semigroup. To obtain an efficient end-to-end quantum
algorithm, we construct a multilevel scheme that works naturally with
multilevel discretizations and that removes polynomial $h$-dependence from both
the readout and the overall complexity.

\subsection{Contributions}
\label{subsec:contributions}

We analyze the complete prepare--simulate--infer (PSI) pipeline under explicit
state-preparation and operator-access assumptions.  The principal
results are as follows.

\paragraph{Direct observable estimation.}
We first construct three block encodings of the parabolic semigroup
$E_h(T)$: a QSVT polynomial construction, a Gaussian dilation based on
the discrete gradient, and an inverse-Laplace contour construction
implemented through first-order shifted systems.  Under their
respective access models, all three have coherent cost
\begin{equation}
  \widetilde{\bigO}\!\left(\frac{\sqrt T}{h}\right).
  \label{eq:intro-direct-coherent-scale}
\end{equation}
Using these block encodings directly inside amplitude estimation avoids
preparation of the normalized final heat state.

We quantify the remaining inference cost through the effective readout
order $\chi$.  This means that the relevant state or observable
block-encoding normalization scales as $\bigO(h^{-\chi})$.  The cases
$\chi=0$, $\chi=1$, and $\chi=2$ include, respectively, bounded
$L^2$-type readouts, first-derivative quantities such as flux, and
second-order or energy-type quadratic observables.  At a fixed spatial
resolution $h$, the direct estimator for \eqref{readouts} has cost
\begin{equation}
  \widetilde{\bigO}\!\left(
    \frac{\sqrt T}{\epsilon h^{1+\chi}}
  \right).
  \label{eq:intro-direct-fixed-mesh}
\end{equation}
If the particular observable has spatial bias $\bigO(h^2/T)$,
then, in the resolved regime $T\eps\lesssim h_0^2$, choosing
\[
  h=\Theta(\sqrt{T\eps})
\]
gives the conditional end-to-end scaling
\begin{equation}
  \mathcal C_{\rm dir}
  =
  \widetilde{\bigO}\!\left(
    T^{-\chi/2}\epsilon^{-(3+\chi)/2}
  \right).
  \label{eq:intro-direct-final}
\end{equation}
Direct observable estimation therefore removes the exponentially small
final-state success probability, but it does not remove the polynomial dependence on $h$.

\paragraph{Multilevel corrected-resolvent estimator.}
To eliminate this remaining polynomial mesh dependence, we introduce a
multilevel estimator on a nested Galerkin hierarchy.  The key novelty of this construction is to create a
fine--coarse cancellation inside the block-encoded level operator. The idea is to couple two consecutive levels $\ell-1$ and $\ell$, by defining
the correction, \[
  \widehat{\be}_\ell(T)
  :=
  h_\ell^{-2}
  \bigl(
    \by_\ell(T)-P_\ell\by_{\ell-1}(T)
  \bigr),
  \,\;
  \bm\phi_\ell(T)
  :=
  \begin{bmatrix}
    \widehat{\be}_\ell(T)\\
    \by_{\ell-1}(T)
  \end{bmatrix}.
\]
We prove exact identities expressing every compatible linear or
quadratic level difference from \eqref{readouts} as an extended observable  of
$\bm\phi_\ell(T)$.  The associated readout contains explicit factors
$h_\ell^2$, so the expected Galerkin correction scale is present before
measurement and after initial preparation.

The other key ingredient is a target-time inverse-Laplace contour representation, where the joint transfer map
through the corrected resolvent (here $L_\ell= L_{h_\ell}$) becomes
\begin{equation}
  D_\ell(z)
  =
  (L_\ell+zI)^{-1}
  -
  P_\ell(L_{\ell-1}+zI)^{-1}P_\ell^\dagger.
  \label{D-res}
\end{equation}
Parabolic smoothing gives an $\bigO(T^{-1})$ normalization for the
linear transfer row, while the extended readout contributes
$h_\ell^{2-\chi}$.  Optimizing the
amplitude-estimation accuracy allocation over the levels yields the complexity for large time $T$
\begin{equation}
  \mathcal C_{\rm ML}
  =
  \widetilde{\bigO}\!\left(1+\frac{1}{T\epsilon}\right),
  \label{eq:intro-ml-final}
\end{equation}
for every $0\le\chi\le2$.  The additional
dependence on $ h$ is only polylogarithmic and the dependence on $1/T$ can be attributed to parabolic smoothing \cite{Thomee2006}.

\paragraph{A corrected-resolvent access model.}
The central implementation problem is not generic access to the two
inverses in \eqref{D-res}.  If they are block encoded independently and
then subtracted by LCU, their separate normalizations obscure the
fine--coarse cancellation.  We instead show an exact shifted
Ritz--Schur factorization
\[
  D_\ell(z)
  =
  J_\ell^R(z)
  \Sigma_\ell(z)^{-1}
  J_\ell^L(z)^\dagger.
\]
 A concise sufficient access condition is that the two shifted Ritz maps
are block encoded with normalization \(\bigO(h_\ell)\), while the shifted
Schur complement is block encoded with normalization
\(\bigO(1+|z_k|h_\ell^2)\) and satisfies the uniform lower bound
\[
  s_{\min}(\Sigma_\ell(z_k))\ge c_\Sigma>0.
\]
QSVT inversion and block-encoding composition then give
\begin{equation}
  \alpha_{D_\ell(z_k)}
  =
  \bigO(h_\ell^2),
  \label{eq:intro-D-access}
\end{equation}
over the selected contour nodes, and therefore the implementation cost has no polynomial
dependence on \(h_\ell^{-1}\).

The overall exponential reduction of mesh dependence assumes that the compatible
level states \(\lvert\by_h(0)\rangle\) and readout states or observables can be
prepared coherently with cost polylogarithmic in \(h_\ell^{-1}\).

The corrected-resolvent access model is explicit for the nested
sine--Galerkin hierarchy considered below: the constant-coefficient
resolvent difference is a diagonal multiplier supported on the newly
added modes. We also give a non-Fourier construction for a uniform dyadic
$P_1$ hierarchy in one dimension when the diffusion coefficient is
piecewise constant on an aligned coarsest mesh. Energy-orthogonal
midpoint details, an explicit discrete Green matrix, and a stable
resolvent identity provide the required normalization and selected cost.
For other structured finite-element hierarchies,
\cref{prop:D-from-shifted-ritz-schur} is a sufficient access condition.
Constructing comparable oracles in higher dimensions or on general
unstructured meshes is a subject of ongoing work.

It is important to point out that the principal improvement is in the dependence on spatial resolution.
At fixed $T\ge1$, the direct implementations analyzed here have cost
$\widetilde O(\epsilon^{-1}h^{-(1+\chi)})$, whereas the multilevel
estimator has cost $\widetilde O(\epsilon^{-1})$ for $0\le\chi\le2$,
under the stated access assumptions, with only polylogarithmic
dependence on $h^{-1}$.
Keeping $h$ and $\epsilon$ separate makes this improvement explicit.

\subsection{Relation to prior work}
\label{subsec:related}
Linden, Montanaro, and Shao
\cite{LindenMontanaroShao2022} give a detailed comparison of classical
and quantum algorithms for estimating regional heat content.
Their accelerated random-walk method with amplitude estimation
already achieves $\widetilde O(\epsilon^{-1})$ complexity for that
target. Our contribution concerns mesh-dependent normalization
in Galerkin observable estimation, including derivative and quadratic
readouts. Under the corrected-resolvent access assumptions,
fine--coarse cancellation removes the polynomial mesh factors
before measurement.

Quantum linear systems and block encodings form the standard algebraic
toolkit for discretized PDEs
\cite{HarrowHassidimLloyd2009,ChildsKothariSomma2017,
CostaEtAl2022Optimal,GilyenSuLowWiebe2019}.
Montanaro and Pallister emphasized systematic end-to-end accounting for
quantum finite element methods \cite{MontanaroPallister2016}.
Deiml and Peterseim obtained optimal linear-functional estimation for
elliptic problems through a BPX-preconditioned block encoding
\cite{DeimlPeterseim2025}.  The elliptic companion of the present work
uses multilevel corrected Green's operators for \emph{both}  linear and quadratic
observables \cite{Li2026EndToEndElliptic}.

For nonunitary dynamics associated with PDEs, relevant approaches include time-marching
quantum ODE solvers, QSVT matrix functions, linear combinations of
Hamiltonian simulations, Schr\"odingerization, and dilation methods
\cite{BerryChildsOstranderWang2017,Krovi2023,FangLinTong2023,
GilyenSuLowWiebe2019,AnLiuLin2023LCHS,
AnChildsLin2023NonUnitary,JinLiuYu2024Schrodingerization,
Li2025MomentDilation}.
Our contour construction belongs to the
linear-combination-of-inverses framework for matrix functions
\cite{TakahiraOhashiSogabeUsuda2022Contour} and uses sectorial
inverse-Laplace quadrature
\cite{LopezFernandezPalenciaSchadle2006,
WeidemanTrefethen2007}.
Its role here is target-time transfer evaluation rather than
construction of a time-uniform reduced model. A complementary observable-driven approach to nonunitary dynamics is
the randomized LCHS framework of
\cite{YangLiu2025ObservableDriven}, which estimates final-time
observable expectations directly through an unbiased Monte Carlo
estimator.  The present construction instead uses deterministic
Galerkin fine--coarse corrections to reduce the block-encoding
normalization before amplitude estimation.

The multilevel principle originates from
\cite{Heinrich2001,Giles2008,Giles2015}.
Quantum-accelerated MLMC for stochastic differential equations
accelerates a probabilistic estimator on each level
\cite{AnLindenLiuMontanaroShaoWang2021}.
The present hierarchy is deterministic: the resource reduced by
fine--coarse coupling is the block-encoding normalization of the
level-transfer map.  The shifted Ritz--Schur construction is related to
variational multigrid, hierarchical bases, prewavelets, and
subspace-correction methods
\cite{Hackbusch1985,BramblePasciakXu1990,Xu1992,
Yserentant1986,Oswald1994,Stevenson1998,
FloaterQuakReimers2000}.

Contour-integral time discretizations for parabolic equations have a
substantial classical history. Sheen, Sloan, and Thom\'ee developed
parallel methods based on contour representations and Laplace-transform
quadrature \cite{SheenSloanThomee1999,SheenSloanThomee2003}.
McLean and Thom\'ee studied iterative solution of the resulting
complex-shifted positive-definite systems \cite{McLeanThomee2011}.
At each quadrature node, an independent shifted elliptic problem is
solved. This links the present construction to the multilevel quantum
elliptic algorithm in \cite{Li2026EndToEndElliptic}.
  The hyperbolic contour used here
is closely related to this classical construction, while our
strip-analytic quadrature estimate and fixed-target-time scaling follow
the more specialized sectorial inverse-Laplace framework of
\cite{LopezFernandezPalencia2004,
LopezFernandezPalenciaSchadle2006}.  Our contour is chosen with a strictly
positive vertex and is scaled by \(T^{-1}\), which is convenient for estimating an observable at
one prescribed target time.

\subsection{Organization and conventions}

\Cref{sec:setup} introduces the parabolic model, the mass-reduced
Galerkin discretization, the Gram factorization, the observable access
model, and the positive-time spatial error estimates.
\Cref{sec:direct} develops the QSVT, Gaussian-dilation, and
contour-resolvent semigroup implementations and their direct
observable-estimation costs.
\Cref{sec:multilevel} constructs the nested hierarchy, joint scaled
correction, target-time contour map, shifted Ritz--Schur access model,
and multilevel complexity.  The appendices collect the contour representation, its approximation, 
 and complexity.

Let $T$ denote the target time in fixed nondimensional units, and let
$\lambda_*>0$ be the fixed Dirichlet spectral-gap lower bound. We state
the multilevel bounds for $T\ge1$; constants may depend on $\lambda_*$.
Constants depending on fixed
coefficient bounds, spatial dimension, mesh regularity, elliptic
regularity, and fixed physical input norms are absorbed into
$\bigO(\cdot)$.  The symbol $\epsilon$ denotes the final additive error,
and polylogarithmic factors are suppressed by
$\widetilde{\bigO}$. For nonnegative quantities \(A\) and \(B\), we write \(A\asymp B\), as is common in finite element analysis, when there exist constants \(c,C>0\), independent of the discretization level and the asymptotic parameters under consideration, such that \(cB\le A\le CB\), i.e., $A=\Theta(B)$.

\section{Parabolic PDE and Galerkin discretization}
\label{sec:setup}
\subsection{Model problem}
\label{subsec:heat-model}

We consider the parabolic PDE \cite{evans2022partial}
\begin{equation}
    \partial_t u(t,x)+\mathcal L u(t,x)=0,
    \,\;
    \mathcal L u=-\nabla\cdot(a(x)\nabla u),
    \,\;
    u(0,x)=u_0(x),
    \label{eq:heat-pde}
\end{equation}
on a bounded polyhedral domain
\[
    \Omega\subset\R^d,
    \,\; d\in\{1,2,3\},
\]
with representative length $\Lbox$, and homogeneous Dirichlet boundary
conditions $u(t,x)=0$ for $x\in\partial\Omega$. We assume uniform elliptic condition on $\mathcal L$ where the coefficient satisfies
\begin{equation}\label{eq:uniform-ellipticity}
     0<\cmin\le a(x)\le \cmax<\infty,
\end{equation}
and that the associated elliptic equation has $H^2 $ regularity.

A key route to study general solution properties and numerical
approximation is the following sesquilinear extension:
On the complexification of \(H_0^1(\Omega)\), we define
\begin{equation}
  \mathrm a(u,v)
  :=
  \int_\Omega
  a(x)\nabla u(x)\cdot\overline{\nabla v(x)}\,\dd x,
  \,\;
  u,v\in H_0^1(\Omega;\mathbb C),
  \label{eq:heat-bilinear}
\end{equation}
with the convention that the form is linear in its first argument.
The solution $u(t,x)$ is then characterized by the weak form: find
$u(t,\cdot)\in H_0^1(\Omega;\C)$, for $t>0$, such that
\cite{Thomee2006}
\begin{equation}\label{weak-form}
  \mathrm a(u(t),v)
  +\int_\Omega \partial_tu(t,x)\,\overline{v(x)}\,\dd x=0,
  \qquad v\in H_0^1(\Omega;\C).
\end{equation}

\subsection{Galerkin semidiscretization in reduced coordinates}
\label{subsec:heat-galerkin}

Let $V_h\subset H_0^1(\Omega)$ be a conforming $P_1$ or $Q_1$ finite element space on a quasi-uniform, shape-regular mesh of size indicated by $h$.  The number of degrees of freedom, which usually determines the cost of a classical algorithm, can be roughly estimated to be
\begin{equation}
    N_h:=\dim V_h\asymp (\Lbox/h)^d,
    \label{eq:Nh}
\end{equation}
where $\Lbox$ is the diameter of the domain $\Omega$.  
To derive a finite-dimensional approximation, let $\{\varphi_j\}_{j=1}^{N_h}$ be the nodal basis in $V_h$ and write
\[
    u_h(t,x)=\sum_{j=1}^{N_h}q_j(t)\varphi_j(x).
\]
For finite elements the basis functions are piecewise polynomials.
The same Galerkin formulation applies to spectral spaces, with a
trigonometric basis in place of the nodal basis.

Enforcing \eqref{weak-form} for all test functions in $V_h$ gives the
Galerkin semidiscretization,
\begin{equation}
  \mathsf  M_h\dot{\bm q}_h(t)+\mathsf K_h\bm q_h(t)=0,
    \label{eq:raw-heat}
\end{equation}
where
\[
    (\mathsf M_h)_{ij}=\int_\Omega \varphi_j(x)\varphi_i(x)\,\dd x,
    \,\;
    (\mathsf K_h)_{ij}=\int_\Omega a(x)\nabla\varphi_j(x)\cdot\nabla\varphi_i(x)\,\dd x.
\]

For quantum encoding, it is useful to work in coordinates whose
Euclidean norm equals the physical $L^2$ norm. Define the mass-reduced
coordinate vector by
\begin{equation}
    \by_h:=\mathsf M_h^{1/2}\bm q_h .
    \label{eq:y-reduced}
\end{equation}
Then we arrive at a system of ordinary differential equations \cite{Thomee2006}, 
\begin{equation}
    \dot\by_h(t)+L_h\by_h(t)=0,
    \,\;
    L_h:=\mathsf M_h^{-1/2} \mathsf K_h \mathsf M_h^{-1/2},
    \label{eq:heat-reduced}
\end{equation}
This is the semidiscrete approximation in mass-reduced coordinates.
On a quasi-uniform mesh, $h^{-d}\mathsf M_h$ is uniformly conditioned.
Given efficient sparse-entry access to this scaled matrix, QSVT provides
its inverse square root at polylogarithmic cost. Restoring the factor
$h^{-d/2}$ and composing with stiffness access constructs $L_h$ with
the appropriate mesh-dependent normalization.

Noticeably, the reduced Euclidean norm coincides with the physical norm:
\begin{equation}
    \norm{\by_h}_2= \norm{u_h}_{L^2(\Omega)} .
    \label{eq:L2-l2-scaling}
\end{equation}
This is the scaling used by the amplitude-encoded quantum state
\[
    |\by_h\rangle=\by_h/\norm{\by_h}_2.
\]
The mass reduction removes the artificial grid-volume scaling of the
nodal vector: $\|\by_h(0)\|$ is the physical norm
$\|u_h(0)\|_{L^2(\Omega)}$. It does not remove the mesh dependence of
derivative readouts or of block-encoding normalizations; these are
accounted for separately below.

\subsection{Gram factorization and operator scales}
\label{subsec:gram}

The stiffness matrix has a Gram structure.  Let $\mathsf G_h$ be the element-gradient or quadrature-gradient matrix, and let $\mathsf  W_h \succ 0$ contain the quadrature weights and coefficient $a(x)$.  Then
\begin{equation}
    \mathsf  K_h=\mathsf  G_h^\dagger \mathsf  W_h \mathsf G_h.
    \label{eq:stiffness-gram}
\end{equation}
Consequently, we arrive at a natural factorization
\begin{equation}
    L_h=A_h^\dagger A_h,
    \,\;
    A_h:=\mathsf W_h^{1/2} \mathsf  G_h \mathsf M_h^{-1/2}.
    \label{eq:A-factor}
\end{equation}
The matrix $A_h$ is the coefficient-weighted discrete gradient in
mass-reduced coordinates. The factors $\mathsf W_h^{1/2}$ and
$\mathsf G_h$ are local, whereas $\mathsf M_h^{-1/2}$ is generally dense
for a consistent finite-element mass matrix. Efficient access to the
factors nevertheless gives a block encoding of their product, and
\begin{equation}
    \norm{A_h}=\bigO(h^{-1}),
    \,\;
    \norm{L_h}=\bigO(h^{-2}).
    \label{eq:A-L-scales}
\end{equation}
The coefficient, local-sparsity, and fixed-dimensional constants are absorbed according to the convention above. 

\subsection{Observables and spatial accuracy}
\label{subsec:heat-error}

We consider two classes of scalar outputs. A basic class of linear
observables consists of bounded $L^2$ functionals, 
\begin{equation}
   Q^{\rm lin}(T)=\int_\Omega \overline{g(x)}\,u(T,x)\,\dd x,
    \,\; g\in L^2(\Omega),
    \label{eq:Qlin-cont}
\end{equation}
with the regional heat content $\int_S u(T,x)\dd x$ of \cite{LindenMontanaroShao2022} corresponding to $g=\mathbf 1_S$.  In reduced coordinates,
\begin{equation}
    Q_h^{\rm lin}(T)=\bg_h^\dagger\,\by_h(T),
    \label{eq:Qlin-disc}
\end{equation}
where $\bg_h$ collects the reduced-coordinate coefficients of the $L^2$ projection of $g$.  
Derivative-dependent linear readouts can instead be defined by
$q(v)=\int_\Omega \overline{b(x)}\cdot\nabla v(x)\,\dd x$, with
$b\in L^2(\Omega)^d$. Their reduced-coordinate Riesz vectors satisfy
$\|\bg_h\|=O(h^{-1})$ by the inverse estimate. Restricting the same
functional to nested spaces gives the compatibility used below.
Efficient preparation of these readout vectors remains an explicit
access assumption, and their continuum bias requires a separate estimate.

Meanwhile \emph{quadratic observables} are quadratic forms, e.g., for some $S\subset \Omega$
\begin{equation}
    Q^{\rm quad}(T)= \int_{S} |u(T,x)|^2\,\dd x
    \,\;\text{or}\,\;
    Q^{\rm quad}(T)=\int_{S} a(x) | \nabla u(T,x)|^2 dx .
    \label{eq:Qquad-cont}
\end{equation}

After the Galerkin subspace projection \eqref{eq:heat-reduced}, and in reduced coordinates, we express them as a quadratic form with a Hermitian matrix $\calM_h$,
\begin{equation}
    Q_h^{\rm quad}(T)=\by_h(T)^\dagger \calM_h\,\by_h(T),
    \label{eq:quad-Qh}
\end{equation}
\cref{eq:Qlin-disc,eq:quad-Qh} together allow a general quadratic function for the readout. 
It is worthwhile noting that 
in the latter case in \eqref{eq:Qquad-cont}, $M_h$ has an $\bigO(h^{-2})$ scaling which may require more amplitude estimation steps.

Standard semidiscrete Galerkin theory gives the spatial error estimates
needed below
\cite{BrambleSchatzThomeeWahlbin1977,Thomee2006,BrennerScott2008}.
For $u_0\in H^2(\Omega)\cap H_0^1(\Omega)$ and initialization
$u_h(0)=\Pi_hu_0$, where $\Pi_h$ is the $L^2$ projection onto $V_h$,
the standard smooth-data estimate is
\begin{equation}
  \sup_{t\ge0}
  \norm{u(t,\cdot)-u_h(t,\cdot)}_{L^2(\Omega)}
  \le
  Ch^2\norm{u_0}_{H^2(\Omega)}.
  \label{eq:heat-smooth-error}
\end{equation}

For general $u_0\in L^2(\Omega)$, with $u_h(0)=\Pi_hu_0$, parabolic
smoothing gives a more precise bound \cite{Thomee2006} 
\begin{equation}
  \norm{u(T,\cdot)-u_h(T,\cdot)}_{L^2(\Omega)}
  \le
  C\min\left\{1,\frac{h^2}{T}\right\}
  \norm{u_0}_{L^2(\Omega)},
  \,\; T>0.
  \label{eq:heat-operator-error}
\end{equation}
For bounded $L^2$ linear readouts, and bounded quadratic forms on
$L^2$, \eqref{eq:heat-operator-error} gives observable bias
$O(h^2/T)$ for fixed physical input and readout norms. Thus, in the
resolved regime, one may choose
\begin{equation}
 h=\Theta\!\left(\min\{h_0,\sqrt{T\eps}\}\right).
 \label{eq:hL-choice}
\end{equation}
For compatible smooth initial data, \eqref{eq:heat-smooth-error}
also permits $h=O(\sqrt\eps)$ uniformly in time for these bounded
readouts. This latter choice is not a uniform short-time conclusion
for general $L^2$ data. Derivative-dependent and energy observables
require their own spatial-bias estimates; their block-encoding
normalization order alone does not establish an $O(h^2/T)$ bias.

\section{Direct estimation of heat observables}
\label{sec:direct}

A natural quantum pipeline would first apply the parabolic semigroup approximation \eqref{eq:heat-reduced},
postselect a state proportional to the final solution, and then measure
the desired observable.  For a dissipative equation, this organization
can be inefficient.  Postselection prepares the normalized state
$\lvert\by_h(T)\rangle$, but its success probability \eqref{pT}
 may decay exponentially with the final time.  We instead place the
semigroup directly inside the observable-estimation circuit, so that no
normalized final heat state is prepared. The same observation has been made in \cite{KharaziAlkadriMandadapuWhaley2026ReactionRates} for estimating reaction rates.  

By setting
\begin{equation}
  E_h(T):=e^{-TL_h},
  \label{eq:heat-semigroup-direct}
\end{equation}
we rewrite the discrete observables in
\cref{eq:Qlin-disc,eq:quad-Qh} as
\begin{align}
  Q_h^{\rm lin}(T)
  &=
  \bg_h^\dagger E_h(T)\by_h(0),
  \label{eq:direct-linear-map}\\
  Q_h^{\rm quad}(T)
  &=
  \by_h(0)^\dagger
  E_h(T)^\dagger\calM_hE_h(T)
  \by_h(0).
  \label{eq:direct-quadratic-map}
\end{align}

\subsection{Observable estimation from a semigroup block encoding}
\label{subsec:single-access}

We assume efficient state-preparation oracles
\begin{equation}
  O_y\lvert0\rangle
  =
  \lvert\by_h(0)\rangle
  :=
  \frac{\by_h(0)}{\norm{\by_h(0)}},
  \,\;
  O_g\lvert0\rangle
  =
  \lvert\bg_h\rangle
  :=
  \frac{\bg_h}{\norm{\bg_h}},
  \label{ass:data-access}
\end{equation}
where $O_g$ is required only for a linear observable.  Let us also keep track of their magnitude
\[
  \beta_{y,h}:=\norm{\by_h(0)},
  \,\;
  \beta_{g,h}:=\norm{\bg_h}.
\]
These physical norms are assumed known from the data-loading
procedure.  Likewise, we assume an
$\alpha_{M,h}$-normalized block encoding of $\calM_h$. 
State preparation, controlled state preparation, observable access, and
evaluation of the known scalar norms are assumed to have
$\polylog(N_h)$ cost in the stated input model.

Suppose that $U_{E,h}(T)$ is an
$\alpha_{E,h}$-normalized block encoding of the semigroup in \eqref{eq:heat-semigroup-direct}
\begin{equation}
  \bigl(\langle0^a|\otimes I\bigr)
  U_{E,h}(T)
  \bigl(|0^a\rangle\otimes I\bigr)
  = E_h(T)/\alpha_{E,h} +
  \bigO(\epsilon),
  \label{eq:semigroup-block-encoding}
\end{equation}
 Since $E_h(T)$ is a contraction, one may take
$\alpha_{E,h}=\bigO(1)$ in the constructions below. 

\begin{proposition}[Observable estimation from a semigroup block encoding]
\label{prop:direct-observable-from-E}
Given the access model above, the linear quantity
$Q_h^{\rm lin}(T)$ can be estimated to additive accuracy $\eps$ using
\begin{equation}
  \widetilde{\bigO}\!\left(
    \frac{
      \alpha_{E,h}\beta_{g,h}\beta_{y,h}
    }{\eps}
  \right)
  \label{eq:linear-E-query-cost}
\end{equation}
controlled uses of the semigroup block encoding $U_{E,h}(T)$ and the state-preparation
oracles.

The quadratic quantity $Q_h^{\rm quad}(T)$ can be estimated to additive
accuracy $\eps$ using
\begin{equation}
  \widetilde{\bigO}\!\left(
    \frac{
      \alpha_{E,h}^2\alpha_{M,h}\beta_{y,h}^2
    }{\eps}
  \right)
  \label{eq:quadratic-E-query-cost}
\end{equation}
uses of the composed semigroup--observable block encoding.  Each such
use requires two calls to $U_{E,h}(T)$ and one call to the block
encoding of $\calM_h$, up to constant ancilla overhead.
\end{proposition}

\begin{proof}
For the linear quantity, the block-encoding identity gives
\begin{equation}
  \bigl\langle0^a,\bg_h\bigr|
  U_{E,h}(T)
  \bigl|0^a,\by_h(0)\bigr\rangle
  =
  \frac{
    Q_h^{\rm lin}(T)
  }{
    \alpha_{E,h}\beta_{g,h}\beta_{y,h}
  }
  +\bigO(\epsilon).
  \label{eq:linear-block-matrix-element}
\end{equation}
The real and imaginary parts of this transition amplitude are obtained
by the Hadamard test applied to coherent preparation of the two
branches
\[
  \lvert0^a,\bg_h\rangle,
  \,\;
  U_{E,h}(T)\lvert0^a,\by_h(0)\rangle
\]
\cite{CleveEtAl1998QuantumAlgorithmsRevisited,
GilyenPoremba2022Fidelity}.
Applying amplitude estimation to the resulting ancilla-measurement
probability gives the query count
\eqref{eq:linear-E-query-cost}.  

For the quadratic quantity, define the pure-state density matrix
\[
  \rho_{y,h}
  :=
  \lvert\by_h(0)\rangle
  \langle\by_h(0)\rvert
\]
and the effective Hermitian observable $\mathcal O_h^{\rm quad}(T)
  :=
  E_h(T)^\dagger\calM_hE_h(T).$
Then
\begin{equation}
  Q_h^{\rm quad}(T)
  =
  \beta_{y,h}^2
  \operatorname{Tr}
  \left(
    \rho_{y,h}\mathcal O_h^{\rm quad}(T)
  \right).
  \label{eq:quadratic-rall-form}
\end{equation}
Block-encoding composition gives
$\mathcal O_h^{\rm quad}(T)$ with normalization
\(
  \alpha_{E,h}^2\alpha_{M,h}.
\)
Applying the block-encoded observable-estimation procedure of
\cite{Rall2020} gives \eqref{eq:quadratic-E-query-cost}.
\end{proof}

The reduced-coordinate scaling removes artificial grid-volume factors,
but the physical readout may still become stronger under mesh
refinement.  We quantify this dependence through the readout orders
\begin{equation}
  \beta_{g,h}
  \le C_g h^{-\chi_{\rm lin}},
  \,\;
  \alpha_{M,h}
  \le C_M h^{-\chi_{\rm quad}},
  \,\;
  0\le\chi_{\rm lin},\chi_{\rm quad}\le2,
  \label{eq:direct-readout-orders}
\end{equation}
with mesh-independent constants $C_g$ and $C_M$.  The input norm
$\beta_{y,h}$ is treated as a fixed physical quantity.  We write $\chi$ for the corresponding
readout order.

\Cref{prop:direct-observable-from-E} separates the inference problem
from the dynamical implementation.  Once a block encoding of $E_h(T)$
is available, the linear and quadratic quantities follow from,
respectively, transition-amplitude estimation and block-encoded
expectation estimation.  The remaining task is therefore to construct
$U_{E,h}(T)$ with favorable normalization and coherent cost. 

Next we present three realizations of $U_{E,h}$. All three use a
block encoding of the first-order factor $A_h$ in
\[
  L_h=A_h^\dagger A_h.
\]
The QSVT construction applies an even singular-value polynomial,
the Gaussian construction simulates a Hermitian dilation, and the
contour construction inverts complex-shifted first-order systems.
The contour formulation will also provide the dynamical primitive for
the multilevel estimator in \cref{sec:multilevel}.

All internal block-encoding, polynomial-approximation, quadrature, and
Hamiltonian-simulation errors below are chosen as fixed fractions of
the final tolerance after multiplication by the known physical
prefactors.  This affects only logarithmic factors, so the same symbol
$\eps$ is used throughout the discussions.

\subsection{QSVT implementation of the parabolic semigroup}
\label{subsec:qsvt}

Assume a block encoding of $A_h/\alpha_A$, where
$L_h=A_h^\dagger A_h$ and $\alpha_A\ge\norm{A_h}$, and set
\begin{equation}
  \tau:=T\alpha_A^2.
  \label{eq:tau-qsvt}
\end{equation}
If $A_h=U\Sigma V^\dagger$ is a singular-value decomposition, then
\[
  E_h(T)=V e^{-T\Sigma^2}V^\dagger.
\]
Consequently, an even singular-value transformation of $A_h/\alpha_A$
with scalar function $f_\tau(x)=e^{-\tau x^2}$ implements the desired
semigroup on the primal register.

To construct the polynomial, apply Corollary~64 of
\cite{GilyenSuLowWiebe2019} (arXiv version) to
$F_{\tau/2}(y)=e^{-(\tau/2)(1-y)}$ on $[-1,1]$. Let
$p_{\tau/2,\eps}$ approximate this function uniformly to error $\eps$,
and define
\begin{equation}
  q_{\tau,\eps}(x):=p_{\tau/2,\eps}(1-2x^2).
  \label{eq:qsvt-even-heat-polynomial}
\end{equation}
Then $q_{\tau,\eps}$ is even and
\[
  \max_{x\in[-1,1]}
  |q_{\tau,\eps}(x)-e^{-\tau x^2}|\le\eps,
  \qquad
  \max_{x\in[-1,1]}|q_{\tau,\eps}(x)|\le1+\eps.
\]
Its degree is
\begin{equation}
  \deg(q_{\tau,\eps})
  =\bigO\!\left(
    \sqrt{\max\{\tau,\log(1/\eps)\}\log(1/\eps)}
  \right).
  \label{eq:gslw-exponential-degree}
\end{equation}
Applying QSVT to $q_{\tau,\eps}/[2(1+\eps)]$ yields an
$\bigO(1)$-normalized block encoding of $E_h(T)$ with operator error
$\bigO(\eps)$; the fixed factor two ensures the standard bounded
real-polynomial admissibility condition.

\begin{proposition}[QSVT block encoding of the parabolic semigroup]
\label{prop:qsvt-heat}
Given a block encoding of $A_h/\alpha_A$, QSVT produces an
$\bigO(1)$-normalized block encoding of
$E_h(T)=e^{-T A_h^\dagger A_h}$ with operator error $\bigO(\eps)$
using
\begin{equation}
  d_{\rm qsvt}
  =\bigO\!\left(
    \alpha_A\sqrt{T\log(1/\eps)}+\log(1/\eps)
  \right)
  \label{eq:qsvt-degree}
\end{equation}
queries to the block encoding of $A_h$ and its adjoint. For
$\alpha_A=\Theta(h^{-1})$, the coherent query complexity is
\begin{equation}
  \mathcal T_{\rm QSVT}(h,T,\eps)
  =\widetildeO\!\left(1+\alpha_A\sqrt T\right)
  =\widetildeO\!\left(1+\frac{\sqrt T}{h}\right).
  \label{eq:qsvt-cost}
\end{equation}
In the regime $T\ge1$ and $h\le1$, this is
$\widetildeO(\sqrt T/h)$.
\end{proposition}

\subsection{Gaussian dilation from discrete-gradient access}
\label{subsec:gaussian-rep}

The second implementation starts from a block encoding of the first-order
factor $A_h$ in \eqref{eq:A-factor}, with normalization
$\alpha_A\ge\norm{A_h}$ and $\alpha_A=\Theta(h^{-1})$.  Its analytic
starting point is the Hubbard--Stratonovich identity
\begin{equation}
  e^{-Tx^2}
  =\frac{1}{\sqrt{4\pi T}}
   \int_{\R}e^{-\omega^2/(4T)}e^{\ii \omega x}\,\dd \omega,
  \,\; T>0.
  \label{eq:gaussian-identity}
\end{equation}
To turn the oscillatory factor into a Hamiltonian simulation, introduce only
at this stage the Hermitian dilation and the projection onto its primal
block,
\begin{equation}
  H_{A_h}:=
  \begin{bmatrix}
    0&A_h^\dagger\\
    A_h&0
  \end{bmatrix},
  \,\;
  \Pi_h:=\begin{bmatrix}I&0\end{bmatrix}.
  \label{eq:H-A}
\end{equation}
Since
\begin{equation}
  H_{A_h}^2
  =\begin{bmatrix}
      L_h&0\\
      0&A_hA_h^\dagger
    \end{bmatrix},
  \label{eq:H-A-square}
\end{equation}
functional calculus applied to \eqref{eq:gaussian-identity} gives the exact
block formula
\begin{equation}
  E_h(T)
  =\Pi_h e^{-T H_{A_h}^2}\Pi_h^\dagger
  =\frac{1}{\sqrt{4\pi T}}
   \int_{\R}e^{-\omega^2/(4T)}
   \Pi_h e^{\ii \omega H_{A_h}}\Pi_h^\dagger\,\dd \omega.
  \label{eq:heat-gaussian}
\end{equation}
Thus the nonunitary parabolic semigroup is a Gaussian linear combination of
unitary Hamiltonian evolutions similar to that in  \cite{KharaziAlkadriMandadapuWhaley2026ReactionRates}.  The generator has the first-order scale
$\norm{H_{A_h}}=\norm{A_h}=\Theta(h^{-1})$, while the Gaussian measure is
concentrated on times $|\omega|=\widetildeO(\sqrt T)$.

\begin{proposition}[Gaussian-dilation block encoding]
\label{prop:gaussian-be}
Under block-encoding access to $A_h/\alpha_A$, truncation and discretization
of \eqref{eq:heat-gaussian}, followed by LCU and optimal Hamiltonian
simulation of $H_{A_h}$, produces an $\bigO(1)$-normalized block encoding
of $E_h(T)$ with operator error $\bigO(\eps)$ and coherent query cost
\begin{equation}
  \mathcal T_{\rm Gauss}(h,T,\eps)
  =\widetildeO\!\left(1+\alpha_A\sqrt T\right)
  =\widetildeO\!\left(1+\frac{\sqrt T}{h}\right).
  \label{eq:gaussian-cost}
\end{equation}
\end{proposition}

\begin{proof}[Proof sketch]
Choose
\begin{equation}
  s_{\max}=\bigO\!\left(\sqrt{T\log(1/\eps)}\right)
  \label{eq:gaussian-cutoff}
\end{equation}
so that the Gaussian tail is within the allocated operator error.  Poisson
summation shows that a trapezoidal spacing
\begin{equation}
  \Delta s
  =\Theta\!\left(
    \left[\alpha_A+\sqrt{\log(1/\eps)/T}\right]^{-1}
  \right)
  \label{eq:gaussian-spacing}
\end{equation}
controls aliasing uniformly for
$\operatorname{spec}(H_{A_h})\subset[-\alpha_A,\alpha_A]$.  The resulting
Gaussian LCU weights are nonnegative and have total mass
$1+\bigO(\eps)$ \cite{ChildsWiebe2012,BerryChildsCleveKothariSomma2015}.

The Gaussian coefficient state can also be prepared with polylogarithmic
cost. Start from a uniform superposition over the retained nodes $s_j$,
compute $e^{-s_j^2/(8T)}$ reversibly, and use this value as a controlled
rotation amplitude. The marked-ancilla probability is
$\Omega(\sqrt T/s_{\max})=\Omega(1/\sqrt{\log(1/\eps)})$.
Coherent amplitude amplification therefore prepares amplitudes
proportional to the square roots of the Gaussian weights with only
polylogarithmic overhead. Arithmetic precision and amplification errors
are included in the allocated operator-error budget. 
A coherently selected simulation with maximum time $s_{\max}$ has query cost
$\widetildeO(1+\alpha_A s_{\max})$ by QSP or qubitization
\cite{LowChuang2017QSP,LowChuang2019Qubitization}.  Logarithmic factors from
the cutoff, quadrature, and simulation precision are absorbed into
\eqref{eq:gaussian-cost}.  Related Gaussian transmutation and dilation
representations appear in
\cite{JinMaZuazua2026Transmutation,KharaziAlkadriMandadapuWhaley2026ReactionRates}.
\end{proof}

\subsection{A contour--based implementation }
\label{subsec:single-contour}

A third implementation follows from the inverse Laplace transform 
\begin{equation}
  E_h(T)
  =\frac{1}{2\pi\ii}
    \int_{\Gamma_T}e^{zT}(zI+L_h)^{-1}\,\dd z,
  \label{eq:single-inverse-laplace}
\end{equation}
where the Bromwich contour $\Gamma_T$ is deformed into the left half-plane
without crossing the poles $-\operatorname{spec}(L_h)$ while making $e^{zT}$ decay rapidly. The contour geometry is shown in
\cref{fig:heat-contour-geometry}.  Each shifted inverse is realized
through the first-order factor \(A_h\).

\begin{lemma}[Sectorially stable inverse-Laplace quadrature]
\label{lem:single-contour-quadrature}
Let $T\ge1$.  There are constants
$C,c,c_0,C_0>0$ and $\delta_0\in(0,\pi/2)$, independent of $h$, $T$, and
the truncation index $n_{\rm c}$, such that for every integer
$n_{\rm c}\ge3$ one can choose conjugate-symmetric shifts and weights
\begin{equation}
  z_{-k}=\overline{z_k},
  \,\;
  w_{-k}=\overline{w_k},
  \,\; -n_{\rm c}\le k\le n_{\rm c},
  \label{eq:single-contour-symmetry}
\end{equation}
with
\begin{equation}
  |\arg z_k|\le\pi-\delta_0,
  \,\;
  \frac{c_0}{T}\le |z_k|\le\frac{C_0n_{\rm c}}{T},
  \label{eq:single-contour-shift-range}
\end{equation}
such that the quadrature error obeys the bound,
\begin{equation}
  \left\|
    E_h(T)-
    \sum_{k=-n_{\rm c}}^{n_{\rm c}}
      w_k(L_h+z_kI)^{-1}
  \right\|
  \le C\exp\!\left(-c\frac{n_{\rm c}}{\log n_{\rm c}}\right).
  \label{eq:single-contour-quadrature-error}
\end{equation}
In addition, the quadrature coefficients have bounds
\begin{equation}
  \sum_{k=-n_{\rm c}}^{n_{\rm c}}
    \frac{|w_k|}{|z_k|}\le C,
  \,\;
  \sum_{k=-n_{\rm c}}^{n_{\rm c}}
    \frac{|w_k|}{\sqrt{|z_k|}}
    \le\frac{C}{\sqrt T}.
  \label{eq:single-contour-weight-bounds}
\end{equation}
\end{lemma}

Consequently, it is sufficient to choose
\begin{equation}
  n_{\rm c}
  =\bigO\!\left(
     \log(1/\eps)\log\log(1/\eps)
   \right)
  \label{eq:single-contour-node-count}
\end{equation}
to make the quadrature error $\bigO(\eps)$.

The lemma is a nonoptimized specialization of hyperbolic sinc quadrature for
sectorial inverse Laplace transforms
\cite{LopezFernandezPalenciaSchadle2006,WeidemanTrefethen2007}.  Its proof is
deferred to \cref{app:single-contour-proofs} in a form that also applies to the multilevel analysis in the next section.

\begin{theorem}[Factorized contour block encoding]
\label{thm:single-contour-block-encoding}
Suppose $L_h=A_h^\dagger A_h$ and that $A_h/\alpha_A$ has a block encoding,
with $\alpha_A\ge\norm{A_h}$.  Let $T\ge1$ and choose $n_{\rm c}$ according
to \eqref{eq:single-contour-node-count}. There is an
$\bigO(1)$-normalized block encoding of $E_h(T)$ with coherent query cost
\begin{equation}
  \mathcal T_{\rm cont}(h,T,\eps)
  =\bigO\!\left(
    \left[1+\alpha_A\sqrt T+\sqrt{n_{\rm c}}\right]
    \log\!\frac{n_{\rm c}(1+\alpha_A\sqrt T)}{\eps}
  \right)
  =\widetildeO\!\left(1+\alpha_A\sqrt T\right).
  \label{eq:single-contour-query-complexity}
\end{equation}
For discrete-gradient access with $\alpha_A=\bigO(h^{-1})$,
\begin{equation}
  \mathcal T_{\rm cont}(h,T,\eps)
  =\widetildeO\!\left(1+\frac{\sqrt T}{h}\right).
  \label{eq:single-contour-cost-h}
\end{equation}
\end{theorem}

The proof begins with the complex-shifted first-order matrix
\begin{equation}
  \mathcal K_h(z)
  =\begin{bmatrix}
      \sqrt z\,I&A_h^\dagger\\
      A_h&-\sqrt z\,I
    \end{bmatrix},
  \label{eq:single-contour-first-order-preview}
\end{equation}
whose square is block diagonal and whose leading inverse block is
$\sqrt z\,(L_h+zI)^{-1}$.  The angular separation in
\eqref{eq:single-contour-shift-range} guarantees a stable branch of
$\sqrt z$ and a first-order condition scale
$\widetildeO(1+\alpha_A\sqrt T)$.  The complete factorization and contour
proof are placed in the appendix.

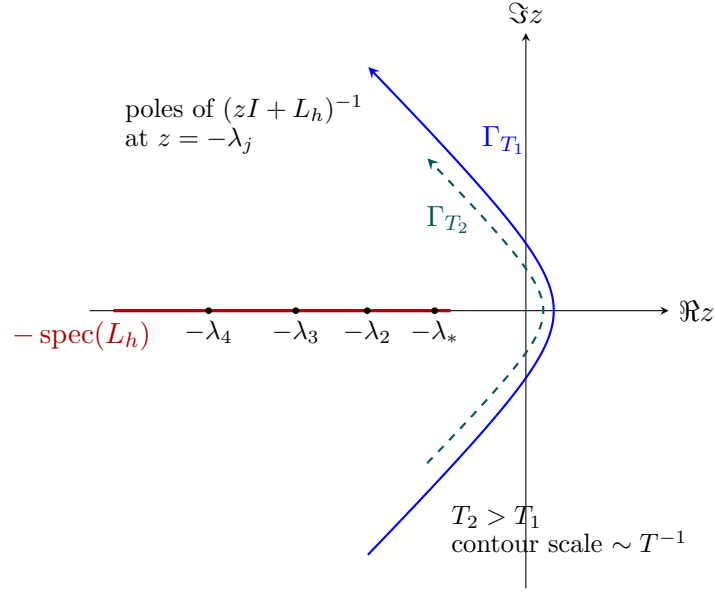
\begin{figure}[t]
\centering
\begin{tikzpicture}[scale=1.05,>=stealth]

  \draw[->] (-5.5,0) -- (1.8,0) node[right] {$\Re z$};
  \draw[->] (0,-3.5) -- (0,3.5) node[above] {$\Im z$};

  \draw[very thick,red!65!black] (-5.2,0) -- (-0.95,0);
  \node[below left,red!65!black] at (-4.6,0)
    {$-\operatorname{spec}(L_h)$};

  \foreach \x/\lab in {
    -1.15/{-\lambda_*},
    -2.0/{-\lambda_2},
    -2.9/{-\lambda_3},
    -4.0/{-\lambda_4}
  }{
    \fill[black] (\x,0) circle (1.2pt);
    \node[below] at (\x,0) {\small $\lab$};
  }

  \draw[thick,blue,->,
        samples=140,domain=-2:2,smooth,variable=\s]
    plot
    ({1.20*(1-0.707106*cosh(\s))},
     {1.20*0.707106*sinh(\s)});

  \draw[thick,teal!70!black,dashed,->,
        samples=140,domain=-2:2,smooth,variable=\s]
    plot
    ({0.75*(1-0.707106*cosh(\s))},
     {0.75*0.707106*sinh(\s)});

  \node[blue,right] at (-0.68,2.15) {$\Gamma_{T_1}$};
  \node[teal!70!black,right] at (-1.38,1.15) {$\Gamma_{T_2}$};

  \node[align=left] at (-3.55,2.35)
    {\small poles of $(zI+L_h)^{-1}$\\[-1mm]
     \small at $z=-\lambda_j$};

  \node[align=left] at (0.55,-2.75)
    {\small $T_2>T_1$\\[-1mm]
     \small contour scale $\sim T^{-1}$};

\end{tikzpicture}
\caption{Schematic hyperbolic contours for the inverse-Laplace
representation.  The contour $\Gamma_T$ intersects the positive real axis
 and its tails extend into the
second and third quadrants.  Increasing the target time contracts the
contour by the factor \(T^{-1}\).}
\label{fig:heat-contour-geometry}
\end{figure}

 Our PDE-specific choices are
the sectorial contour, whose constants are uniform on the positive spectral
half-axis, and the first-order realization of each shifted inverse through
$A_h$, as shown in \cref{fig:heat-contour-geometry}.  This differs from contour-based matrix decomposition
\cite{wang2025quantum}, where the contour argument is used
to decompose a general nonunitary matrix function into implementable
Hermitian components.  Here the numerical contour is discretized directly
into shifted resolvents.  That form is particularly convenient in the next
section, where the same nodes and weights are retained while a single-level
resolvent is replaced by a fine--coarse corrected resolvent.

As outlined in \cite{McLeanThomee2011}, the contour construction also extends naturally to a nonhomogeneous
equation
\(
  \partial_t u+\mathcal L u=f(t),
  \,\;
  u(0)=u_0.
\)
Indeed, after Laplace transformation,
\(
  \widehat u(z)
  =
  (zI+\mathcal L)^{-1}
  \bigl(u_0+\widehat f(z)\bigr),
\)
yielding
\[
  u(T)
  =
  \frac{1}{2\pi\ii}
  \int_{\Gamma_T}
  e^{zT}
  (zI+\mathcal L)^{-1}
  \bigl(u_0+\widehat f(z)\bigr)\,\dd z.
\]
The representation extends to forcing when $\widehat f(z)$ continues
holomorphically to the deformation region and satisfies suitable
sectorial bounds. A quantum implementation additionally requires
coherent access to $u_0+\widehat f(z_k)$, including its node-dependent
normalization and preparation error. These factors enter the LCU and
measurement costs, so the homogeneous complexity bounds do not apply
to general forcing without further assumptions.

\subsection{Observable estimator and end-to-end complexity}
\label{subsec:pipeline}

With the block-encoding of the semigroup $E_h(T)$, we invoke \cref{prop:direct-observable-from-E} and summarize the cost of estimating directly the linear and quadratic functions in the regime where $h \ll 1 $.

\begin{theorem}[Direct observable-estimation cost at a fixed mesh $h$]
\label{thm:direct-eps}
Let $T\ge1$ and $h\le1$. Suppose the state-preparation and
observable-access assumptions of \cref{subsec:single-access} hold, and let
\begin{equation}
  \mathcal T_E(h,T,\eps)
  =\widetildeO\!\left(\frac{\sqrt T}{h}\right)
  \label{eq:unified-semigroup-cost}
\end{equation}
be the coherent cost of any one of the three semigroup implementations.
Then amplitude estimation with constant success probability estimates the
linear and quadratic observables to additive physical accuracy $\eps$ at
costs
\begin{align}
  \calC_{\rm dir}^{\rm lin}(h,T,\eps)
  &=\widetildeO\!\left(
      \frac{\sqrt T}{h\eps}
      \norm{\by_h(0)}\norm{\bg_h}
    \right),  \quad 
  \calC_{\rm dir}^{\rm quad}(h,T,\eps)
  =\widetildeO\!\left(
      \frac{\sqrt T}{h\eps}
      \norm{\by_h(0)}^2\alpha_{M,h}
    \right).
  \label{eq:direct-fixed-h-quadratic}
\end{align}
Equivalently, if the relevant readout scale is bounded by
$C_{\rm obs}h^{-\chi}$, then
\begin{equation}
  \calC_{\rm dir}(h,T,\eps)
  =\widetildeO\!\left(
      C_{\rm obs}\frac{\sqrt T}{\eps h^{1+\chi}}
    \right).
  \label{eq:direct-fixed-h}
\end{equation}
\end{theorem}

\section{The multilevel estimation algorithm}
\label{sec:multilevel}

The direct estimators of \cref{sec:direct} avoid the final-state
normalization penalty, but they still apply a fixed-mesh semigroup inside
every amplitude-estimation call.  To approximate the continuum quantity of
interest to accuracy $\eps$, the mesh width must be sufficiently small,
which leads to a polynomial dependence on $h^{-1}$ in the coherent
semigroup implementation and, for derivative-dependent quantities, in the
readout normalization.  We therefore introduce a hierarchy of nested
discretizations
\[
  h_0>h_1>\cdots>h_L=h,
\]
where the finest level $L$ is chosen to meet the required spatial
accuracy.  The purpose of the multilevel construction is to move the
fine--coarse cancellation inside the quantum map before measurement.

Let $Q_\ell(T)$ denote the quantity of interest computed on level $\ell$.
The exact telescoping identity is
\begin{equation}
  Q_L(T)
  =
  Q_0(T)
  +
  \sum_{\ell=1}^{L}\Delta_\ell(T),
  \,\;
  \Delta_\ell(T)
  :=
  Q_\ell(T)-Q_{\ell-1}(T).
  \label{eq:heat-telescope}
\end{equation}
The coarsest term captures the large-scale part of the response, while
$\Delta_\ell(T)$ represents the additional information introduced when
the mesh is refined from $h_{\ell-1}$ to $h_\ell$.

In classical multilevel Monte Carlo, consecutive discretizations are
coupled so that the variance of $\Delta_\ell$ decreases with the level
\cite{Heinrich2001,Giles2008,Giles2015}.   The hierarchy considered here is
deterministic.  Its quantum analogue of variance reduction is a smaller
block-encoding normalization for the level-correction map.
This distinction is important because the number of
amplitude-estimation queries is proportional to the normalization of the
quantity being estimated.  If the fine- and coarse-grid outputs were
estimated in separate circuits and subtracted only afterward, both
circuits would still pay the normalization of a full solution-level
observable, and no multilevel gain would result.  The fine--coarse
cancellation must therefore be implemented coherently, before
measurement, so that the quantum circuit directly encodes the smaller
quantity $\Delta_\ell(T)$.

For standard $P_1/Q_1$ Galerkin discretizations and sufficiently regular
solutions, the interlevel solution correction has physical scale
$\bigO(h_\ell^2)$.  Once this factor is exposed inside the block encoding,
a readout of order $\chi$ has level scale
$\bigO(h_\ell^{2-\chi})$, up to the time-dependent semigroup factor
derived below.  Coarse levels carry the larger contributions, whereas the
more highly resolved levels contribute progressively smaller corrections
and therefore require fewer coherent repetitions.  The multilevel
algorithm exploits this balance rather than repeatedly evaluating the full
finest-grid observable.

The construction has three ingredients.  First, the nested Galerkin
hierarchy is expressed in compatible mass-reduced coordinates.  Second, a
two-level variable gives exact linear and quadratic representations of the
level correction.  Third, a target-time contour formula expresses the
scaled correction through the shifted corrected resolvent in \cref{eq:joint-corrected-resolvent-z}.

\subsection{Nested Galerkin hierarchy and mass-scaled transfer}
\label{subsec:hierarchy}

Rather than using a discretization with a fixed resolution, we consider a hierarchy of discretizations. Let
\begin{equation}
  V_0\subset V_1\subset\cdots\subset V_L\subset H_0^1(\Omega),
  \,\;
  h_\ell=2^{-\ell}h_0,
  \label{eq:nested-spaces}
\end{equation}
be conforming \(P_1\) or \(Q_1\) finite element spaces on a shape-regular,
quasi-uniform refinement hierarchy.  To distinguish the finite element mass
matrix from the quadratic readout \(\calM_\ell\), write
\(\mathsf M_\ell\) and \(\mathsf K_\ell\) for the mass and stiffness
matrices:
\begin{equation}
  (\mathsf M_\ell)_{ij}
  :=\left(\varphi_j^{(\ell)},\varphi_i^{(\ell)}\right)_{L^2(\Omega)},
  \,\;
  (\mathsf K_\ell)_{ij}
  :=\mathrm a\left(\varphi_j^{(\ell)},\varphi_i^{(\ell)}\right).
  \label{eq:level-mass-stiffness}
\end{equation}
The semidiscrete equation and its mass-reduced form are
\begin{equation}
  \mathsf M_\ell\dot{\bq}_\ell+\mathsf K_\ell\bq_\ell=0,
  \,\;
  \by_\ell:=\mathsf M_\ell^{1/2}\bq_\ell,
  \,\;
  L_\ell:=\mathsf M_\ell^{-1/2}\mathsf K_\ell
            \mathsf M_\ell^{-1/2},
  \label{eq:level-reduced-generator}
\end{equation}
so that
\begin{equation}
  \dot\by_\ell(t)=-L_\ell\by_\ell(t),
  \,\;
  \by_\ell(t)=e^{-tL_\ell}\by_\ell(0).
  \label{eq:level-reduced-heat}
\end{equation}
Notice that we have compressed notations from $\bm y_\ell$ to $\bm y_{h_\ell}$, from $L_\ell$ to $L_{h_\ell}$, and so on.

Let \(I_\ell\) be the nodal coefficient matrix of the natural inclusion
$V_{\ell-1}\hookrightarrow V_\ell$. We use the exact Galerkin hierarchy 
\begin{equation}
  \mathsf M_{\ell-1}=I_\ell^\dagger\mathsf M_\ell I_\ell,
  \,\;
  \mathsf K_{\ell-1}=I_\ell^\dagger\mathsf K_\ell I_\ell.
  \label{eq:raw-galerkin-identities}
\end{equation}
The corresponding reduced-coordinate prolongation is
\begin{equation}
  P_\ell
  :=\mathsf M_\ell^{1/2}I_\ell\mathsf M_{\ell-1}^{-1/2}:
  \C^{N_{\ell-1}}\longrightarrow\C^{N_\ell}.
  \label{eq:mass-scaled-prolongation}
\end{equation}
It represents the same finite element function on the two levels and satisfies
\begin{equation}
  P_\ell^\dagger P_\ell=I,
  \,\;
  L_{\ell-1}=P_\ell^\dagger L_\ell P_\ell.
  \label{eq:heat-galerkin-compatibility}
\end{equation}
Thus \(P_\ell\) is the natural inclusion in \(L^2\)-orthonormal
coordinates, and \(P_\ell^\dagger\) is the associated \(L^2\) projection.
This is the standard variational transfer pair in multigrid and
subspace-correction methods
\cite{Hackbusch1985,BramblePasciakXu1990,Xu1992}.

The Dirichlet spectral gap is uniform in the level.  The Rayleigh--Ritz
principle and Poincar\'e inequality give
\begin{equation}
  \lambda_{\min}(L_\ell)
  =\min_{0\ne v_\ell\in V_\ell}
    \frac{\mathrm a(v_\ell,v_\ell)}{\norm{v_\ell}_{L^2(\Omega)}^2}
  \ge \frac{\cmin}{C_{\rm P}^2}
  =:\lambda_*>0,
  \label{eq:level-generator-spectrum}
\end{equation}
while the inverse estimate gives
\(\lambda_{\max}(L_\ell)=\bigO(h_\ell^{-2})\). 

\begin{assumption}[Compatible data, readouts, and level preparation]
\label{ass:compatible-data}
The level data and readouts satisfy
\begin{equation}
  \by_{\ell-1}(0)=P_\ell^\dagger\by_\ell(0),
  \,\;
  \bg_{\ell-1}=P_\ell^\dagger\bg_\ell,
  \,\;
  \calM_{\ell-1}=P_\ell^\dagger\calM_\ell P_\ell.
  \label{eq:compatible-data}
\end{equation}
For every level, a data oracle prepares
\(
  \lvert\by_\ell(0)\rangle
  =\by_\ell(0)/\norm{\by_\ell(0)}
\)
with cost $\widetilde O(1)$ in powers of $h_\ell^{-1}$, and the norm
$\|\by_\ell(0)\|$ is known and uniformly bounded in the hierarchy.
Linear readout states are likewise efficiently prepared with known
norms, which may grow according to \eqref{eq:direct-readout-orders}. 
\end{assumption}

\begin{lemma}[Positive-time interlevel smoothing]
\label{lem:interlevel-smoothing}
Assume elliptic \(H^2\) regularity and the Galerkin hierarchy above.  Then one has the error bound
\begin{equation}
  \left\|
    e^{-tL_\ell}-P_\ell e^{-tL_{\ell-1}}P_\ell^\dagger
  \right\|
  \le C\min\left\{1,\frac{h_\ell^2}{t}\right\},
  \,\; t>0.
  \label{eq:interlevel-operator}
\end{equation}
Consequently, under \cref{ass:compatible-data},
\begin{equation}
  \norm{\by_\ell(T)-P_\ell\by_{\ell-1}(T)}
  \le C\frac{h_\ell^2}{T}\norm{\by_\ell(0)},
  \,\; T\ge1.
  \label{eq:interlevel-final-time}
\end{equation}
\end{lemma}

The proof is given in \cref{app:interlevel-smoothing}.  The scaled
correction is therefore \(\bigO(T^{-1})\) in operator norm.

\subsection{Joint level dynamics and exact correction observables}
\label{subsec:joint-dynamics}

We now fix two consecutive levels, and define the fine--coarse correction and the
interlevel defect by
\begin{equation}
  \be_\ell(t):=\by_\ell(t)-P_\ell\by_{\ell-1}(t),
  \,\;
  B_\ell:=L_\ell P_\ell-P_\ell L_{\ell-1}.
  \label{eq:coarse-and-error}
\end{equation}
The Galerkin identities \eqref{eq:heat-galerkin-compatibility} imply that
\begin{equation}
  B_\ell=(I-P_\ell P_\ell^\dagger)L_\ell P_\ell,
  \,\;
  P_\ell^\dagger B_\ell=0.
  \label{eq:defect-blocks}
\end{equation}
Direct differentiation gives the exact two-grid  coupled equations
\begin{equation}
  \dot\be_\ell=-L_\ell\be_\ell-B_\ell\by_{\ell-1},
  \,\;
  \dot\by_{\ell-1}=-L_{\ell-1}\by_{\ell-1}.
  \label{eq:unscaled-joint-dynamics}
\end{equation}
Related all-at-once and space--time multigrid formulations for parabolic
equations are discussed in
\cite{HortonVandewalle1995,GanderNeumueller2016SpaceTimeMG}.

Under compatible initial data,
\begin{equation}
  \by_{\ell-1}(0)=P_\ell^\dagger\by_\ell(0),
  \,\;
  \be_\ell(0)=(I-P_\ell P_\ell^\dagger)\by_\ell(0),
  \label{eq:joint-initial-data}
\end{equation}
and motivated by \eqref{eq:interlevel-final-time}, we consider 
\begin{equation}
  \widehat\be_\ell(t):=h_\ell^{-2}\be_\ell(t),
  \,\;
  \bm\phi_\ell(t):=
  \begin{bmatrix}
    \widehat\be_\ell(t)\\[1mm]
    \by_{\ell-1}(t)
  \end{bmatrix}.
  \label{eq:joint-state-phi}
\end{equation}

Then from \cref{eq:unscaled-joint-dynamics}, we arrive at a coupled system
\begin{equation}
  \frac{\dd}{\dd t}\bm\phi_\ell(t)
  =-\widehat{L}_\ell\bm\phi_\ell(t),
  \,\;
  \widehat{L}_\ell:=
  \begin{bmatrix}
    L_\ell&h_\ell^{-2}B_\ell\\[1mm]
    0&L_{\ell-1}
  \end{bmatrix}.
  \label{eq:explicit-joint-dynamics}
\end{equation}

Our first key observation is that $\bm\phi_\ell$ encodes sufficient information to estimate the observable differences. 
\begin{lemma}[Exact linear and quadratic correction identities]
\label{lem:exact-dynamic-corrections}
Under \cref{ass:compatible-data}, let
\begin{equation}
  Q_\ell^{\rm lin}(T):=\bg_\ell^\dagger\by_\ell(T),
  \,\;
  Q_\ell^{\rm quad}(T):=\by_\ell(T)^\dagger\calM_\ell\by_\ell(T),
  \label{eq:hatted-level-observables}
\end{equation}
and \(\Delta_\ell:=Q_\ell-Q_{\ell-1}\).  Define
\begin{equation}
  \widehat{\bg}_\ell:=
  \begin{bmatrix}h_\ell^2\bg_\ell\\0\end{bmatrix},
  \,\;
  \widehat{\calM}_\ell:=
  \begin{bmatrix}
    h_\ell^4\calM_\ell&h_\ell^2\calM_\ell P_\ell\\[1mm]
    h_\ell^2P_\ell^\dagger\calM_\ell&0
  \end{bmatrix}.
  \label{eq:extended-correction-observables}
\end{equation}
Then
\begin{equation}
  \Delta_\ell^{\rm lin}(T)
  =\widehat{\bg}_\ell^\dagger\bm\phi_\ell(T),
  \,\;
  \Delta_\ell^{\rm quad}(T)
  =\bm\phi_\ell(T)^\dagger\widehat{\calM}_\ell\bm\phi_\ell(T).
  \label{eq:extended-correction-identities}
\end{equation}
\end{lemma}

\begin{proof}
Use
\(
  \by_\ell(T)=P_\ell\by_{\ell-1}(T)+h_\ell^2\widehat\be_\ell(T)
\)
and the readout compatibility in \eqref{eq:compatible-data}.  The linear
identity is immediate.  Expanding the quadratic form and subtracting the
embedded coarse--coarse term gives the second identity.
\end{proof}

\subsection{Target-time contour representation}
\label{subsec:joint-contour}

The scaled joint generator in \eqref{eq:explicit-joint-dynamics} is not a
favorable object for direct simulation.  Although its spectrum is the union
of the two diagonal spectra, its Euclidean scale satisfies
\begin{equation}
  \norm{L_\ell}+\norm{L_{\ell-1}}
  =\bigO(h_\ell^{-2}),
  \,\;
  \norm{h_\ell^{-2}B_\ell}
  =\bigO(h_\ell^{-4}).
  \label{eq:scaled-joint-generator-obstruction}
\end{equation}
We therefore construct only the target-time input--output map required by
the observables, rather than simulate the joint generator as an
unstructured ODE system.

To start, we define the scaled correction map and the full two-row transfer map by
\begin{align}
  \widehat{D}_\ell(t)
  &:=
  h_\ell^{-2}
  \left(
    e^{-tL_\ell}
    -
    P_\ell e^{-tL_{\ell-1}}P_\ell^\dagger
  \right),
  \label{eq:scaled-correction-transfer-map}\\
  \widehat E_\ell(t)
  &:=
  \begin{bmatrix}
    \widehat{D}_\ell(t)\\[1mm]
    e^{-tL_{\ell-1}}P_\ell^\dagger
  \end{bmatrix},
  \,\;
  \bm\phi_\ell(t)
  =
  \widehat E_\ell(t)\by_\ell(0).
  \label{eq:scaled-joint-transfer-map}
\end{align}
Only the upper row is needed for a linear correction.  The lower row is
only needed for the cross terms in a quadratic correction.  But neither
estimator prepares a normalized joint state.

Next we give the integral expression for $\widehat E_\ell(t)$.
For $z\notin-\operatorname{spec}(L_j)$ define $R_j(z)$ (the
definition of $D_j(z)$ also requires
$z\notin-\operatorname{spec}(L_{j-1})$):
\begin{equation}
  R_j(z):=(L_j+zI)^{-1},
  \,\;
  D_j(z)
  :=
  R_j(z)-P_jR_{j-1}(z)P_j^\dagger,
  \quad j\ge1.
  \label{eq:joint-corrected-resolvent-z}
\end{equation}

We now incorporate these representations to the contour method from \cref{subsec:single-contour}.
\begin{lemma}[Laplace transform of the two-level transfer map]
\label{lem:joint-transfer-resolvent}
For \(\Re z>0\),
\begin{align}
  \int_0^\infty e^{-zt}\widehat{D}_\ell(t)\,\dd t
  &=
  h_\ell^{-2}D_\ell(z),
  \label{eq:correction-transfer-resolvent-z}\\
  \int_0^\infty e^{-zt}\widehat E_\ell(t)\,\dd t
  &=
  \begin{bmatrix}
    h_\ell^{-2}D_\ell(z)\\[1mm]
    R_{\ell-1}(z)P_\ell^\dagger
  \end{bmatrix}.
  \label{eq:joint-transfer-resolvent-z}
\end{align}
Consequently,
\begin{equation}
  \widehat{D}_\ell(T)
  =
  \frac{1}{2\pi\ii}
  \int_{\Gamma_T}
  e^{zT}h_\ell^{-2}D_\ell(z)\,\dd z .
  \label{eq:correction-transfer-contour-z}
\end{equation}
\end{lemma}

\begin{proof}
Apply the identity
\[
  \int_0^\infty e^{-zt}e^{-tL_j}\,\dd t
  =
  (L_j+zI)^{-1},
  \,\; \Re z>0,
\]
to the two rows in
\eqref{eq:scaled-correction-transfer-map}--%
\eqref{eq:scaled-joint-transfer-map}.  The last formula follows by inverse
Laplace transformation and deformation to the admissible sectorial
contour.
\end{proof}

It can be shown that these required analytic input has normalization, 
\begin{equation}
  \norm{D_\ell(z)}
  \le
  C_\delta
  \frac{h_\ell^2}{1+|z|h_\ell^2},
  \,\;
  \norm{R_j(z)}
  \le
  \frac{C_\delta}{\lambda_*+|z|},
  \,\;
  |\arg z|\le\pi-\delta.
  \label{eq:sectorial-two-level-estimates}
\end{equation}
The first estimate combines a shifted two-level Aubin--Nitsche argument
for \(|z|h_\ell^2\lesssim1\) with the elementary large-shift bound.  Its
proof is given in \cref{lem:sectorial-corrected-resolvent}.

\begin{theorem}[Contour block encoding of the scaled correction]
\label{thm:scaled-correction-contour}
Assume \eqref{eq:sectorial-two-level-estimates}.  For \(T\ge1\) and
\(0<\eps<e^{-2}\), there is a conjugate-symmetric hyperbolic quadrature
with nodes and weights
\[
  \mathcal Z_{T,\eps}
  =
  \{z_k:-n_{\rm c}\le k\le n_{\rm c}\},
  \,\;
  z_{-k}=\overline{z_k},
  \quad
  \omega_{-k}=\overline{\omega_k},
\]
such that
\begin{equation}
  |\arg z_k|\le\pi-\delta_0,
  \,\;
  \frac{c_0}{T}
  \le |z_k|
  \le
  \frac{C_0n_{\rm c}}{T},
  \label{eq:admissible-shift-range}
\end{equation}
and
\begin{align}
  \left\|
    \widehat{D}_\ell(T)
    -
    \sum_{k=-n_{\rm c}}^{n_{\rm c}}
      \omega_k h_\ell^{-2}D_\ell(z_k)
  \right\|
  &\le
  \frac{C}{T}
  \exp\!\left(
    -c\frac{n_{\rm c}}{\log n_{\rm c}}
  \right),
  \label{eq:correction-transfer-quadrature-error}\\
  \sum_{k=-n_{\rm c}}^{n_{\rm c}}
  \frac{|\omega_k|}
       {1+|z_k|h_r^2}
  &\le
  \frac{C}{T},
  \,\; 1\le r\le L.
  \label{eq:correction-transfer-weight-bound}
\end{align}
It is thus sufficient to take
\(
  n_{\rm c}
  =
  \bigO\!\left(
    \log(1/\eps)\log\log(1/\eps)
  \right).
\)

Suppose that \(D_\ell(z_k)\) has a selected block encoding with
\begin{equation}
  \alpha_{D_\ell(z_k)}
  \le C_D h_\ell^2.
  \label{eq:joint-D-access-target-z}
\end{equation}
Here the selected encoding is required to approximate
$D_\ell(z_k)$ to absolute operator error at most
$\eta h_\ell^2$, uniformly in $k$, with cost polylogarithmic in
$\eta^{-1}$. The resulting error after multiplication by
$h_\ell^{-2}$ and contour summation is $\bigO(\eta/T)$. 
Then \(\widehat{D}_\ell(T)\) from \cref{eq:correction-transfer-contour-z} has a block encoding with
\begin{equation}
  \alpha_{\widehat{D},\ell}
  =
  \bigO(T^{-1}).
  \label{eq:correction-transfer-block-normalization}
\end{equation}
\end{theorem}

The proof is provided in \cref{proof-contour-correction}. The quadrature estimate is a Banach-space-valued specialization of
sectorial inverse-Laplace quadrature
\cite{LopezFernandezPalencia2004,
LopezFernandezPalenciaSchadle2006}.
Its quantum implementation uses the linear-combination-of-inverses
framework of \cite{TakahiraOhashiSogabeUsuda2022Contour}.  The same
target-time nodes and weights are used on every spatial level. Here the new
ingredient is that each nodewise inverse is encoded only after the
Galerkin fine--coarse cancellation has been exposed.

\subsection{Explicit corrected-resolvent access models}
\label{subsec:shifted-ritz-schur}

The preceding theorem reduces the algorithmic problem to constructing
\(D_\ell(z_k)\) at the normalization in
\eqref{eq:joint-D-access-target-z}.  We first verify this requirement
directly for a nested sine--Galerkin hierarchy and then give a
Ritz--Schur sufficient condition for structured finite-element
discretizations.

\paragraph{Nested sine--Galerkin hierarchy.}
Consider
\(
  \mathcal L=-\nu\Delta,
  \; 
  \nu>0,
\)
on \(\Omega=(0,L_{\rm box})^d\) with homogeneous Dirichlet boundary
conditions.  For
\(\boldsymbol j=(j_1,\ldots,j_d)\in\mathbb N^d\), set
\[
  \varphi_{\boldsymbol j}(x)
  :=
  \left(\frac{2}{L_{\rm box}}\right)^{d/2}
  \prod_{q=1}^{d}
  \sin\left(\frac{\pi j_qx_q}{L_{\rm box}}\right),
  \,\;
  \lambda_{\boldsymbol j}
  =
  \nu\left(\frac{\pi}{L_{\rm box}}\right)^2
  |\boldsymbol j|_2^2 .
\]
Set $n_\ell=2^\ell n_0$,
$\mathcal I_\ell=\{1,\ldots,n_\ell\}^d$, and
$V_\ell=\operatorname{span}\{\varphi_{\boldsymbol j}:
\boldsymbol j\in\mathcal I_\ell\}$.
The corresponding resolution parameter is
$h_\ell=L_{\rm box}/n_\ell$, and $N_\ell=n_\ell^d$.
The sine modes form an orthonormal Galerkin basis, so the mass matrix
is the identity and
\begin{equation}
  L_\ell
  =
  \diag_{\boldsymbol j\in\mathcal I_\ell}
  (\lambda_{\boldsymbol j}).
  \label{eq:sine-diagonal-generator}
\end{equation}
The prolongation \(P_\ell\) retains the coarse modal coefficients and
pads the newly added coefficients with zero. 
Due to the Fourier diagonalization,  one has the explicit formula, 
\begin{equation}
  D_\ell(z)
  =
  \diag_{\boldsymbol j\in\mathcal I_\ell}
  \left(
    \frac{\mathbf 1_{\mathcal I_\ell\setminus\mathcal I_{\ell-1}}(\boldsymbol j)}
         {\lambda_{\boldsymbol j}+z}
  \right).
  \label{eq:spectral-corrected-resolvent}
\end{equation}
Every newly added mode satisfies
\(\lambda_{\boldsymbol j}=\Theta(h_\ell^{-2})\).  Since the contour stays
a fixed angle away from the negative real axis, one directly  gets the bound \cref{eq:joint-D-access-target-z}. The oracle is a standard diagonal-function block encoding.  Reversible
arithmetic evaluates 
\((\lambda_{\boldsymbol j}+z_k)^{-1}\), followed by a controlled
rotation and uncomputation.  If the input is represented in grid
coordinates, this diagonal multiplier is conjugated by the
tensor-product quantum sine transform.  For fixed dimension, the selected
cost is \(\polylog(N_\ell,1/\eps)\).  The same diagonal construction
implements the coarse row
\(e^{-TL_{\ell-1}}P_\ell^\dagger\) with \(\bigO(1)\) normalization.
This explicit spectral-access model also permits a direct block encoding
of $e^{-TL_\ell}$ by evaluating $e^{-T\lambda_{\boldsymbol j}}$
reversibly. 

QFT-based PDE circuits are developed in
\cite{LubaschKikuchiWrightMcKeever2025FourierPDE}, and related spectral
transform circuits have been demonstrated on the Quantinuum H1-1
processor \cite{WrightEtAl2024Wave}.  These works do not implement the
present corrected-resolvent contour construction, but they realize its
principal transform primitive.

\paragraph{Shifted Ritz--Schur factorization.}
For a non-spectral finite-element hierarchy, one can choose a detail coordinate map
\(W_\ell:\C^{m_\ell}\to\C^{N_\ell}\) such that
\[
  \Phi_\ell:=[\,P_\ell\ \ W_\ell\,]
\]
is invertible. Let \(K_\ell(z):=L_\ell+zI\) be a shifted elliptic operator, and define
\[
  G_\ell(z):=P_\ell^\dagger K_\ell(z)W_\ell,
  \,\;
  F_\ell(z):=W_\ell^\dagger K_\ell(z)P_\ell,
  \,\;
  H_\ell(z):=W_\ell^\dagger K_\ell(z)W_\ell.
\]
Then the Galerkin projection gives 
\begin{equation}
  \Phi_\ell^\dagger K_\ell(z)\Phi_\ell
  =
  \begin{bmatrix}
    K_{\ell-1}(z)&G_\ell(z)\\
    F_\ell(z)&H_\ell(z)
  \end{bmatrix}.
  \label{eq:block-K-z}
\end{equation}
For complex \(z\), \(G_\ell(z)\) and \(F_\ell(z)\) need not be
adjoints themselves. Nevertheless, the relation \(G_\ell(z)=F_\ell(z^*)^\dagger\) holds from the symmetry of the contour. 

To express $D_\ell$ in these new coordinates, we introduce the Schur and Ritz complements, 
\begin{align}
  \Sigma_\ell(z)
  &:=
  H_\ell(z)
  -
  F_\ell(z)K_{\ell-1}(z)^{-1}G_\ell(z),
  \label{eq:shifted-Schur-z}\\
  J_\ell^R(z)
  &:=
  W_\ell
  -
  P_\ell K_{\ell-1}(z)^{-1}G_\ell(z),
  \notag\\
  J_\ell^L(z)^\dagger
  &:=
  W_\ell^\dagger
  -
  F_\ell(z)K_{\ell-1}(z)^{-1}P_\ell^\dagger.
  \label{eq:shifted-detail-maps-z}
\end{align}

\begin{lemma}[Shifted Ritz--Schur factorization]
\label{lem:D-factorization-z}
For every admissible shift \(z\),
\begin{equation}
  D_\ell(z)
  =
  J_\ell^R(z)\Sigma_\ell(z)^{-1}
  J_\ell^L(z)^\dagger.
  \label{eq:D-factorization-z}
\end{equation}
\end{lemma}

\begin{proof}
Apply the block-inverse formula to
\(\Phi_\ell^\dagger K_\ell(z)\Phi_\ell\) and subtract the lifted coarse
resolvent \cite{Li2026EndToEndElliptic}.
\end{proof}

\begin{proposition}[Ritz--Schur corrected-resolvent access]
\label{prop:D-from-shifted-ritz-schur}
Suppose, uniformly over the selected level and contour-node registers,
that block encodings of the shifted factors are available with
\begin{equation}
  \alpha_{J_\ell^R(z_k)},
  \alpha_{J_\ell^L(z_k)}
  =
  \bigO(h_\ell),
  \,\;
  \alpha_{\Sigma_\ell(z_k)}
  =
  \bigO(1+|z_k|h_\ell^2),
  \,\;
  s_{\min}(\Sigma_\ell(z_k))
  \ge
  c_\Sigma>0.
  \label{eq:ritz-schur-access-conditions}
\end{equation}
If the selected implementations of these factors have
\(\widetildeO(1)\) cost in powers of \(h_\ell^{-1}\), then
\(D_\ell(z_k)\) has a selected block encoding with
\begin{equation}
  \alpha_{D_\ell(z_k)}
  =
  \bigO(h_\ell^2)
  \label{eq:D-normalization-z}
\end{equation}
and \(\widetildeO(1)\) selected cost in powers of \(h_\ell^{-1}\).
\end{proposition}

\begin{proof}
The Hermitian-dilation inversion of \(\Sigma_\ell(z_k)\) has inverse
normalization \(\bigO(1)\) and query cost
\[
  \widetildeO\!\left(
    \frac{\alpha_{\Sigma_\ell(z_k)}}
         {s_{\min}(\Sigma_\ell(z_k))}
  \right)
  =
  \widetildeO(1+|z_k|h_\ell^2).
\]
Because
\[
  |z_k|
  \le
  \frac{C_0n_{\rm c}}{T},
  \,\;
  n_{\rm c}
  =
  \widetildeO(1),
\]
this cost is \(\widetildeO(1)\) in powers of \(h_\ell^{-1}\) for
\(T\ge1\) and \(h_\ell\le1\).  Composition with the two shifted Ritz
maps proves \eqref{eq:D-normalization-z}.
\end{proof}

\paragraph{A structured finite-element realization.}
\cref{prop:D-from-shifted-ritz-schur} suggests a concrete non-Fourier construction. 
For example, one can consider continuous
$P_1$ elements on uniform dyadic meshes of $\Omega=(0,1)$.
Assume that the diffusion coefficient is uniformly elliptic and
piecewise constant on a fixed coarsest partition, with all refinements
preserving its interfaces. We first construct the unshifted Green
operators explicitly, then obtain the shifted corrections through
a stable resolvent identity.

Let $S_\ell$ inject the new midpoint nodes into the fine nodal
coordinates, and set
\[
  W_\ell=\sqrt{h_\ell}\,\mathsf M_\ell^{1/2}S_\ell.
\]
Its columns represent the energy-scaled midpoint hats.
Because a coarse function has constant derivative on each parent
element $K$, these details are energy-orthogonal to the coarse space
\cite{Yserentant1986,Oswald1994}. Consequently,
\[
  P_\ell^\dagger L_\ell W_\ell=0,
  \qquad
  \Sigma_\ell(0)
  =W_\ell^\dagger L_\ell W_\ell
  =\operatorname{diag}_{K}(2a_K).
\]
The scaled mass matrix $\mathsf M_\ell/h_\ell$ is sparse and uniformly
well-conditioned. Its square root therefore has an efficient
$O(1)$-normalized encoding, giving
$\alpha_{W_\ell}=O(h_\ell)$. The diagonal inverse above then yields
\[
  D_\ell(0)
  =W_\ell\Sigma_\ell(0)^{-1}W_\ell^\dagger,
  \qquad
  \alpha_{D_\ell(0)}=O(h_\ell^2).
\]

The full unshifted inverse also admits direct access in one dimension.
On a mesh $x_j=jh$, $j=0,\ldots,n$, with $nh=1$, write $a_j$
for the coefficient on the $j$th element and define
\[
  r_j=\frac{h}{a_j},
  \qquad
  s_i=\sum_{j=1}^{i}r_j.
\]
The nodal stiffness matrix has the explicit Green matrix
\begin{equation}
  (\mathsf K_h^{-1})_{ij}
  =
  \frac{s_{\min(i,j)}
        \bigl(s_n-s_{\max(i,j)}\bigr)}{s_n},
  \qquad 1\le i,j\le n-1.
  \label{eq:one-dimensional-discrete-green}
\end{equation}
Since the coefficient has a fixed number of pieces, these prefix
sums can be evaluated with polylogarithmic cost in $h^{-1}$ and
the requested precision. The entries are uniformly bounded, so
standard entry-access block encoding gives normalization $O(h^{-1})$
for $\mathsf K_h^{-1}$ \cite{GilyenSuLowWiebe2019}.
In mass-reduced coordinates,
\[
  R_h(0)=L_h^{-1}
  =\mathsf M_h^{1/2}\mathsf K_h^{-1}\mathsf M_h^{1/2}.
\]
The two mass factors each contribute normalization $O(h^{1/2})$,
giving an $O(1)$-normalized encoding of $R_h(0)$ with
polylogarithmic cost. The usual mass-scaled prolongation likewise
has an efficient $O(1)$-normalized encoding.

To introduce the complex shift, define
$C_\ell=P_\ell R_{\ell-1}(0)P_\ell^\dagger$.
Since $P_\ell^\dagger P_\ell=I$, the exact identity
\begin{equation}
  D_\ell(z)
  =
  \bigl(I+zR_\ell(0)\bigr)^{-1}
  D_\ell(0)
  \bigl(I+zC_\ell\bigr)^{-1}
  \label{eq:one-dimensional-shift-transport}
\end{equation}
holds without requiring $R_\ell(0)$ and $C_\ell$ to commute.
For $|\arg z|\le\pi-\delta$, positivity gives
\[
  s_{\min}\bigl(I+zR_\ell(0)\bigr),
  \quad
  s_{\min}\bigl(I+zC_\ell\bigr)
  \ge \sin(\delta/2).
\]
QSVT therefore encodes each outer inverse with $O(1)$ normalization
using $\widetilde O(1+|z|)$ queries. At the contour nodes,
$|z_k|=O(n_{\rm c}/T)$, this cost is polylogarithmic for $T\ge1$.
Composition preserves the central correction's scale:
\[
  \alpha_{D_\ell(z_k)}=O(h_\ell^2),
\]
with absolute error $O(\eta h_\ell^2)$ and polylogarithmic
dependence on $\eta^{-1}$. Extensions of this construction will be persued in separate works.

\subsection{Multilevel complexity}
\label{subsec:ml-complexity}

The coarsest contribution $Q_0(T)$ is fixed dimensional, and the
coarse generator $L_0$ is classically specified by the discretization.  It can be
computed classically when its data and readout entries are classically
available to the requested precision.  Under the state-preparation
oracle model alone, it can instead be estimated quantumly.  The
fixed-dimensional semigroup has a block encoding with normalization
$e^{-\lambda_*T}\le C/T$ for $T\ge1$, so its contribution costs
$\widetilde O(1+(T\eps_0)^{-1})$ at accuracy $\eps_0$.

For $\ell\ge1$, write
\[
  \mathcal V_\ell=P_\ell e^{-TL_{\ell-1}}P_\ell^\dagger.
\]
The selected access assumptions give block encodings with
$\alpha_{\widehat D,\ell}=O(T^{-1})$ and $\alpha_{\mathcal V,\ell}=O(1)$, at
polylogarithmic cost.  For the latter construction we require efficient
$O(1)$-normalized access to $P_\ell$ and to the coarse semigroup row.
The linear correction has estimation scale
\[
 c_\ell^{\rm lin}
 =h_\ell^2\alpha_{\widehat D,\ell}\beta_{g,\ell}\beta_{y,\ell}.
\]
For a quadratic correction, the effective observable on the fine-level
input is
\[
 \mathcal O_{\Delta,\ell}
 =h_\ell^4 \widehat D_\ell(T)^\dagger\calM_\ell \widehat D_\ell(T)
  +h_\ell^2\bigl(
    \widehat D_\ell(T)^\dagger\calM_\ell \mathcal V_\ell
   +\mathcal V_\ell^\dagger\calM_\ell \widehat D_\ell(T)\bigr).
\]
Block-encoding composition and LCU give
\[
 c_\ell^{\rm quad}
 \le \beta_{y,\ell}^2\alpha_{M,\ell}
 \left(h_\ell^4\alpha_{\widehat D,\ell}^2
       +2h_\ell^2\alpha_{\widehat D,\ell}\alpha_{\mathcal V,\ell}\right).
\]
Consequently, under \eqref{eq:direct-readout-orders}, $T\ge1$,
$h_\ell\le1$, and uniformly bounded input norms, both cases satisfy
\[
 c_\ell=\widetilde O(h_\ell^{2-\chi}/T).
\]
Every term in the quadratic correction contains at least one occurrence
of $\widehat D_\ell(T)$; this is how its $T^{-1}$ normalization is retained.
Amplitude estimation to levelwise accuracy $\eps_\ell$ costs
$\widetilde O(1+c_\ell/\eps_\ell)$.

Reserve a fixed fraction $\eps_{\rm ML}$ of the total error budget for
the noncoarse corrections and require
$\sum_{\ell=1}^L\eps_\ell\le\eps_{\rm ML}$.
The Cauchy--Schwarz inequality gives
\[
 \left(\sum_{\ell=1}^L\sqrt{c_\ell}\right)^2
 \le \left(\sum_{\ell=1}^L\frac{c_\ell}{\eps_\ell}\right)
       \left(\sum_{\ell=1}^L\eps_\ell\right).
\]
Equality gives the allocation
\begin{equation}
 \eps_\ell=\eps_{\rm ML}
       \frac{\sqrt{c_\ell}}{\sum_{r=1}^L\sqrt{c_r}},\qquad
 \mathcal C_{\rm corr}
 =\widetildeO\!\left[
 L+\frac{1}{\eps_{\rm ML}}
       \left(\sum_{\ell=1}^L\sqrt{c_\ell}\right)^2
 \right].
 \label{eq:ml-ae}
\end{equation}
Levels with a zero correction bound can be omitted. 
Since
\(
  \sqrt{c_\ell}
  =
  \widetildeO\!\left(
    \frac{h_\ell^{1-\chi/2}}{\sqrt T}
  \right),
\)
the sum reduces to $\sum_{\ell=1}^{L}h_\ell^{1-\chi/2}$, which is
$\bigO(1)$ for $0\le\chi<2$ and $\bigO(L)$ for $\chi=2$.
We summarize the resulting complexity as follows.

\begin{theorem}[Multilevel contour-transfer complexity]
\label{thm:ml-complexity}
Let $Q_L(T)$ be a finest-level discrete linear or quadratic observable,
and assume $T\ge1$, $h_0\le1$, \cref{ass:compatible-data},
the sectorial bounds \eqref{eq:sectorial-two-level-estimates}, and the
readout orders \eqref{eq:direct-readout-orders} with $0\le\chi\le2$.
Suppose that, uniformly over level and contour-node registers,
$D_\ell(z_k)$ has a selected block encoding with normalization
$O(h_\ell^2)$, while $P_\ell$ and
$e^{-TL_{\ell-1}}P_\ell^\dagger$ have $O(1)$-normalized block encodings.
Assume that these operations, the readout encodings, and the controlled
data preparations have cost
$\polylog(h_\ell^{-1},T,\eta^{-1})$
at normalized implementation accuracy $\eta$. Assume also that the
fixed-dimensional coarse generator $L_0$ is classically specified.
Then $Q_L(T)$ can be estimated to additive accuracy $\eps$, with
success probability at least $2/3$, at cost
\begin{equation}
  \mathcal C_{\rm ML}
  =
  \widetildeO\!\left( 1 + 
    \frac{1}{T\eps}
  \right).
  \label{eq:ML-final-costs}
\end{equation}

\end{theorem}

This theorem estimates the discrete observable $Q_L(T)$.  To estimate
a continuum observable $Q(T)$, choose $L$ using a separate bound
$|Q(T)-Q_L(T)|\le\eps/2$ and apply the theorem at tolerance $\eps/2$.
The $L^2$ semigroup error implies an $O(h_L^2/T)$ bias for bounded
$L^2$ linear and quadratic readouts.  Derivative-dependent readouts
require their own observable-bias estimate; the normalization order
$\chi$ alone does not imply that rate.

Contour, arithmetic, block-encoding, and amplitude-estimation errors are
assigned fixed portions of the final error budget.  Amplifying the
levelwise success probabilities and taking a union bound over
\(L=\bigO(\log(1/h))\) levels changes only polylogarithmic factors.

\section{Summary and outlook}
\label{sec:summary}

This paper treats observable estimation, rather than preparation of a
normalized solution state, as the primary quantum task for parabolic PDEs.
For dissipative dynamics, this distinction is important: the solution norm
may decay rapidly even when a physically meaningful linear or quadratic
quantity remains accessible as a matrix element of the semigroup.  We give
three direct implementations of the semigroup, including QSVT, Gaussian dilation,
and contour-resolvent methods, with coherent cost
$\widetildeO(\sqrt{T}/h)$.  These methods avoid the final-state
postselection penalty, but they retain polynomial dependence on the finest
mesh and on derivative-dependent readout normalizations.

Under the selected-access hypotheses, the multilevel construction removes
the remaining polynomial mesh dependence by encoding the fine--coarse
correction coherently before measurement.  A target-time contour
representation reduces each level correction to shifted resolvent problems,
while the Ritz--Schur factorization exposes the cancellation at its natural
two-grid scale.  Combining the resulting inverse-time normalization with
the level-dependent readout scale and optimizing the amplitude-estimation
accuracy across levels gives
\(\widetildeO(1+(T\epsilon)^{-1})\) complexity for both linear and
quadratic observables.  Dependence on the finest mesh is confined to
polylogarithmic factors under the stated selected-access assumptions.

The Fourier hierarchy plays a particularly important practical role.  For
constant-coefficient problems on tensor-product domains, a quantum Fourier
transform, or the corresponding sine or cosine transform for the boundary
conditions, diagonalizes the discrete generator.  The fine--coarse
correction is then represented by a spectral-band selector, and each
contour resolvent becomes a diagonal function of the spectral and contour
registers.  Its implementation requires only the spectral transform,
reversible evaluation of the eigenvalue and contour multiplier, and a
controlled rotation, without an iterative linear solver or a
coefficient-dependent operator-adapted basis. This makes the spectral
case a natural candidate for circuit-level benchmarking.
QFT-based PDE circuits are developed in
\cite{LubaschKikuchiWrightMcKeever2025FourierPDE}, and related
spectral-transform circuits have been demonstrated on Quantinuum H1-1
\cite{WrightEtAl2024Wave}. Implementing the present estimator also
requires controlled arithmetic, LCU preparation, readout access, and
amplitude estimation; the cited demonstrations do not establish the
resources needed for this complete circuit.

Beyond the Fourier setting, we give an explicit construction for a
uniform dyadic \(P_1\) hierarchy in one dimension with uniformly elliptic
diffusion that is piecewise constant on an aligned coarsest mesh.
Energy-orthogonal midpoint details expose the \(O(h_\ell^2)\) scale
of the unshifted correction, while an explicit discrete Green matrix
provides efficient access to the unshifted inverse. A stable resolvent
identity then transfers this construction to the complex contour shifts,
preserving the correction's normalization and polylogarithmic
implementation cost. This example illustrates how classical finite
element structure can supply the operator access needed by the multilevel
algorithm. Extending the construction to higher dimensions and more general
coefficients is underway, and several routes have been identified in
which finite element coordinates and coarse solvers are designed jointly
with their quantum implementations, so that efficient access preserves
the small scale of the fine--coarse correction. These constructions are
delicate, and we leave them to separate work.

The observable-first formulation may extend to polynomial readouts and
other evolution PDEs, including wave and advection--diffusion equations.
The transfer maps, correction scales, and contour representations must
be adapted to each PDE and require a separate analysis. The design
principle is the same: exposing fine--coarse cancellation inside the quantum map
before amplitude estimation.

\section*{Acknowledgments}
This work was supported by NSF Grant DMS-2411120 and DMS-2552687.
\bibliographystyle{plain}
\bibliography{heat}

\appendix
\numberwithin{equation}{section}
\section{Contour-based representations}
\label{app:single-contour-proofs}

The same inverse-Laplace quadrature is used in the single-level and
multilevel constructions.  We use \(z\) for the complex Laplace variable and
\(n_{\rm c}\) for the number of positive contour indices.

\subsection{Sector geometry and a first-order dilation}

The following elementary estimate is standard in sectorial
inverse-Laplace analysis; see
\cite{LopezFernandezPalencia2004,
LopezFernandezPalenciaSchadle2006}
for the sectorial geometry and resolvent bounds underlying the contour
construction.

\begin{lemma}[Elementary sector geometry]
\label{lem:scalar-sector-geometry}
Let \(a\ge0\) and let \(z\in\C\) satisfy
\(|\arg z|\le\pi-\delta\), with \(0<\delta<\pi\).  Then
\begin{equation}
  |a+z|\ge \sin(\delta/2)(a+|z|).
  \label{eq:app-scalar-sector-geometry}
\end{equation}
\end{lemma}

\begin{proof}
Write \(z=re^{\ii\theta}\).  For any fixed \(\theta\), and fixed $a+r>0$ the minimum of
\(|a+re^{\ii\theta}|/(a+r)\) over \(a,r\ge0\), occurs at
\(a=r\) and the minimum value equals \(|\cos(\theta/2)|\) by law of cosine.  Since
$|\theta|\le\pi-\delta$, the minimum is at least $\sin(\delta/2)$.
\end{proof}

\begin{lemma}[First-order realization of a complex-shifted resolvent]
\label{lem:complex-shifted-factorization}
Let \(A\in\C^{m\times n}\), let \(L=A^\dagger A\), and let
\(z\in\C\setminus(-\infty,0]\).  Choose the principal square root
\begin{equation}
  \rho:=\sqrt z,
  \,\; \operatorname{Re}\rho>0,
  \label{eq:app-complex-shift-root}
\end{equation}
and define a dilation
\begin{equation}
  \mathcal K_A(z):=
  \begin{bmatrix}
    \rho I_n&A^\dagger\\
    A&-\rho I_m
  \end{bmatrix}.
  \label{eq:app-complex-shifted-system}
\end{equation}
Let $\Pi_{\rm pr}:=[\,I_n\ \ 0_{n\times m}\,]$ denote the
projection onto the primal block. Then
\begin{equation}
  (L+zI)^{-1}
  =\rho^{-1}\Pi_{\rm pr}\mathcal K_A(z)^{-1}\Pi_{\rm pr}^\dagger.
  \label{eq:app-complex-resolvent-factorization}
\end{equation}
If \(\sigma\) is a singular value of \(A\), the two singular values of
\(\mathcal K_A(z)\) on the corresponding singular-vector subspace are
\begin{equation}
  \left[(\operatorname{Re}\rho)^2+
    \bigl(\sigma\pm|\operatorname{Im}\rho|\bigr)^2\right]^{1/2}.
  \label{eq:app-complex-K-singular-values}
\end{equation}
Additional singular values associated with unmatched null spaces equal
\(|\rho|\).  Consequently, the extreme singular values satisfy
\begin{equation}
  s_{\min}(\mathcal K_A(z))\ge\operatorname{Re}\rho,
  \,\;
  \norm{\mathcal K_A(z)}\le\norm A+|\rho|.
  \label{eq:app-complex-K-basic-bounds}
\end{equation}
If in addition \(|\arg z|\le\pi-\delta\), then
\begin{equation}
  \operatorname{Re}\rho\ge\sin(\delta/2)\sqrt{|z|},
  \,\;
  \kappa(\mathcal K_A(z))
  \le C_\delta\left(1+\frac{\norm A}{\sqrt{|z|}}\right).
  \label{eq:app-complex-K-sector-bound}
\end{equation}
\end{lemma}

\begin{proof}
 
\eqref{eq:app-complex-resolvent-factorization} follows from Schur complement of $\mathcal K_A(z)$ directly.

To determine the singular values, let
\[
  A\bm v=\sigma\bm u,
  \,\;
  A^\dagger\bm u=\sigma\bm v,
  \,\;
  \norm{\bm u}=\norm{\bm v}=1,
\]
where $\bm u\in\C^m$ and $\bm v\in\C^n$ are left and right singular
vectors of $A$ corresponding to the singular value $\sigma$.  The
two-dimensional space
\[
  \operatorname{span}
  \left\{
    \begin{bmatrix}\bm v\\0\end{bmatrix},
    \begin{bmatrix}0\\\bm u\end{bmatrix}
  \right\}
\]
is invariant under $\mathcal K_A(z)$.  Define
\begin{equation}
  \bq_\sigma^\pm
  :=
  \frac{1}{\sqrt2}
  \begin{bmatrix}
    \bm v\\
    \pm\ii\bm u
  \end{bmatrix}.
  \label{eq:complex-K-right-singular-vectors}
\end{equation}
A direct calculation gives
\begin{align}
  \mathcal K_A(z)\bq_\sigma^+
  &=
  (\rho+\ii\sigma)\bq_\sigma^-,
  \label{eq:complex-K-action-plus}\\
  \mathcal K_A(z)\bq_\sigma^-
  &=
  (\rho-\ii\sigma)\bq_\sigma^+.
  \label{eq:complex-K-action-minus}
\end{align}
Hence the two singular values associated with $\sigma$ are
\begin{equation}
  s_\sigma^\pm
  :=
  |\rho\pm\ii\sigma|
  =
  \sqrt{
    (\Re\rho)^2+
    (\Im\rho\pm\sigma)^2
  }.
  \label{eq:complex-K-singular-values}
\end{equation}

The remaining null-space directions are also explicit.  If
$\bm v\in\ker(A)$, then
\[
  \mathcal K_A(z)
  \begin{bmatrix}\bm v\\0\end{bmatrix}
  =
  \rho
  \begin{bmatrix}\bm v\\0\end{bmatrix},
\]
while, if $\bm u\in\ker(A^\dagger)$, then
\[
  \mathcal K_A(z)
  \begin{bmatrix}0\\\bm u\end{bmatrix}
  =
  -\rho
  \begin{bmatrix}0\\\bm u\end{bmatrix}.
\]
These directions therefore contribute the singular value $|\rho|$.

Consequently,
\begin{equation}
  s_{\min}\bigl(\mathcal K_A(z)\bigr)
  \ge \Re\rho,
  \label{eq:complex-K-min-singular-value}
\end{equation}
because every singular value in
\eqref{eq:complex-K-singular-values} is at least $\Re\rho$, and
$|\rho|\ge\Re\rho$.  The same decomposition also gives
\begin{equation}
  \norm{\mathcal K_A(z)}
  =
  \sqrt{
    (\Re\rho)^2+
    \bigl(|\Im\rho|+\norm A\bigr)^2
  }
  \le
  |\rho|+\norm A.
  \label{eq:complex-K-exact-norm}
\end{equation}
Finally, if $|\arg z|\le\pi-\delta$, then
\[
  \Re\rho
  =
  \sqrt{|z|}
  \cos\left(\frac{\arg z}{2}\right)
  \ge
  \sin(\delta/2)\sqrt{|z|},
\]
which proves the sectorial bounds.
\end{proof}

\subsection{Inverse-Laplace quadrature}
\label{app:operator-contour}

The convergence result below adapts the hyperbolic inverse-Laplace
quadrature of \cite{LopezFernandezPalencia2004,
LopezFernandezPalenciaSchadle2006} to a fixed target time and a fixed
contour scale proportional to $T^{-1}$. The Banach-space-valued
formulation and the different algebraic growth exponents are discussed
in \cite[Section~2, Theorem~1, and Remark~2]
{LopezFernandezPalenciaSchadle2006}.
We record the time scaling and coefficient estimates required by the
quantum linear-combination construction.

\begin{lemma}[Operator-valued sectorial inverse-Laplace quadrature]
\label{lem:operator-valued-contour}
Let \(X,Y\) be finite-dimensional Hilbert spaces and define
\begin{equation}
  \mathfrak S_\delta
  :=
  \{z\in\C\setminus\{0\}:|\arg z|<\pi-\delta\},
  \,\;
  0<\delta<\frac{\pi}{2}.
  \label{eq:abstract-sector}
\end{equation}
Suppose that
\(\mathcal R:\mathfrak S_\delta\to\mathcal B(X,Y)\) is holomorphic and
satisfies
\begin{equation}
  \norm{\mathcal R(z)}
  \le
  M|z|^{-\nu},
  \,\;
  0\le\nu\le1.
  \label{eq:abstract-sectorial-potential-bound}
\end{equation}
Assume that the target-time operator has the deformed inverse-Laplace
representation
\begin{equation}
  \mathcal U(T)
  =
  \frac{1}{2\pi\ii}
  \int_{\Gamma_T}
  e^{zT}\mathcal R(z)\,\dd z .
  \label{eq:abstract-inverse-laplace}
\end{equation}

Choose fixed parameters
\begin{equation}
  0<\vartheta<\frac{\pi}{2}-\delta,
  \,\;
  0<d<
  \min\left\{
    \vartheta,\frac{\pi}{2}-\delta-\vartheta
  \right\},
  \,\;
  a_{\rm c}>0,
  \label{eq:hyperbola-fixed-parameters}
\end{equation}
and define
\begin{equation}
  \zeta(w)
  :=
  a_{\rm c}\bigl(1-\sin(\vartheta-\ii w)\bigr),
  \,\;
  \Gamma_T
  :=
  \{\zeta(x)/T:x\in\R\}.
  \label{eq:proof-hyperbolic-contour}
\end{equation}
Then the integral in \eqref{eq:abstract-inverse-laplace} is absolutely
convergent and
\begin{equation}
  \norm{\mathcal U(T)}
  \le
  CMT^{\nu-1}.
  \label{eq:abstract-time-scale}
\end{equation}

For \(n_{\rm c}\ge3\), let
\begin{equation}
  \kappa_{n_{\rm c}}
  :=
  \frac{\log n_{\rm c}}{n_{\rm c}},
  \,\;
  z_k
  :=
  \frac{\zeta(k\kappa_{n_{\rm c}})}{T},
  \label{eq:abstract-contour-nodes}
\end{equation}
and
\begin{equation}
  \omega_k
  :=
  \frac{\kappa_{n_{\rm c}}}{2\pi\ii T}
  e^{\zeta(k\kappa_{n_{\rm c}})}
  \zeta'(k\kappa_{n_{\rm c}}),
  \,\;
  -n_{\rm c}\le k\le n_{\rm c}.
  \label{eq:abstract-contour-weights}
\end{equation}
There exist constants
\(c,C,c_0,C_0,\delta_0>0\), depending only on the fixed contour
parameters, such that
\begin{equation}
  \left\|
    \mathcal U(T)
    -
    \sum_{k=-n_{\rm c}}^{n_{\rm c}}
      \omega_k\mathcal R(z_k)
  \right\|
  \le
  CMT^{\nu-1}
  \exp\left(
    -c\frac{n_{\rm c}}{\log n_{\rm c}}
  \right).
  \label{eq:abstract-contour-error}
\end{equation}
Moreover,
\begin{equation}
  |\arg z_k|
  \le
  \pi-\delta_0,
  \,\;
  \frac{c_0}{T}
  \le
  |z_k|
  \le
  \frac{C_0n_{\rm c}}{T},
  \label{eq:abstract-contour-shift-range}
\end{equation}
and
\[
  z_{-k}=\overline{z_k},
  \,\;
  \omega_{-k}=\overline{\omega_k}.
\]
Finally, for every \(0\le r\le1\),
\begin{equation}
  \sum_{k=-n_{\rm c}}^{n_{\rm c}}
  \frac{|\omega_k|}{|z_k|^r}
  \le
  CT^{r-1}.
  \label{eq:abstract-contour-general-weight-bound}
\end{equation}
In particular,
\begin{equation}
  \sum_k|\omega_k|
  \le\frac{C}{T},
  \,\;
  \sum_k\frac{|\omega_k|}{\sqrt{|z_k|}}
  \le\frac{C}{\sqrt T},
  \,\;
  \sum_k\frac{|\omega_k|}{|z_k|}
  \le C.
  \label{eq:abstract-contour-special-weight-bounds}
\end{equation}
\end{lemma}

\begin{proof}
Because the algorithm evaluates the dynamics at one prescribed target time
\(T\), it is natural to scale a fixed dimensionless contour by \(T^{-1}\).
Indeed, setting \(\xi=Tz\) in
\eqref{eq:abstract-inverse-laplace} gives
\begin{equation}
  \mathcal U(T)
  =
  \frac{1}{2\pi\ii T}
  \int_{\Gamma_1}
  e^\xi\mathcal R(\xi/T)\,\dd\xi.
  \label{eq:fixed-time-scaled-contour}
\end{equation}
After parameterizing \(\xi=\zeta(w)\), the integrand becomes
\begin{equation}
  \mathcal G_T(w)
  :=
  \frac{1}{2\pi\ii T}
  e^{\zeta(w)}
  \zeta'(w)
  \mathcal R\left(\frac{\zeta(w)}{T}\right).
  \label{eq:abstract-parameterized-integrand}
\end{equation}
The potential bound therefore gives
\begin{equation}
  \norm{\mathcal G_T(w)}
  \le
  CMT^{\nu-1}
  e^{\Re\zeta(w)}
  \frac{|\zeta'(w)|}{|\zeta(w)|^\nu}.
  \label{eq:abstract-integrand-strip-bound}
\end{equation}

For \(w=x+\ii y\), one has
\[
  \Re\zeta(x+\ii y)
  =
  a_{\rm c}
  \left(
    1-\sin(\vartheta+y)\cosh x
  \right).
\]
The restrictions in
\eqref{eq:hyperbola-fixed-parameters} ensure that the strip
\(|\Im w|<d\) is mapped into a closed subsector of
\(\mathfrak S_\delta\).  On this strip,
\[
  \Re\zeta(x+\ii y)
  \le
  a_{\rm c}-c\cosh x,
  \,\;
  |\zeta(x+\ii y)|\asymp\cosh x,
  \,\;
  |\zeta'(x+\ii y)|\asymp\cosh x.
\]
Thus \(\mathcal G_T\) is analytic in the strip and integrable on its
boundary, with norm bounded by
\(CMT^{\nu-1}\) times a double-exponentially decaying majorant.
This proves \eqref{eq:abstract-time-scale}.

Fix $0<d'<d$. The trapezoidal rule for a function analytic in
$|\Im w|<d$ and integrable on the lines $\Im w=\pm d'$ has error
\[
 \left\|\int_{\R}\mathcal G_T(x)\,\dd x
       -\kappa\sum_{k\in\mathbb Z}\mathcal G_T(k\kappa)\right\|
 \le \frac{CMT^{\nu-1}}{e^{2\pi d'/\kappa}-1}.
\]
This is the strip-analytic estimate underlying the hyperbolic
quadrature analyses in
\cite{LopezFernandezPalencia2004,LopezFernandezPalenciaSchadle2006}.
The double-exponential majorant also gives
\[
 \kappa\sum_{|k|>n_{\rm c}}
       \|\mathcal G_T(k\kappa)\|
 \le CMT^{\nu-1}\exp(-c e^{n_{\rm c}\kappa}).
\]
Taking $\kappa=\log(n_{\rm c})/n_{\rm c}$ yields
\eqref{eq:abstract-contour-error}. Both estimates hold for
$0\le\nu\le1$, including $\nu=0$.

It remains to record the geometry and coefficient sums used later.
For real \(x\),
\begin{equation}
  |\zeta(x)|
  =
  a_{\rm c}\bigl(\cosh x-\sin\vartheta\bigr).
  \label{eq:hyperbola-modulus}
\end{equation}
Since
\(|k\kappa_{n_{\rm c}}|\le\log n_{\rm c}\),
this gives
\[
  c_0
  \le
  |\zeta(k\kappa_{n_{\rm c}})|
  \le
  C_0n_{\rm c},
\]
which proves the magnitude bounds in
\eqref{eq:abstract-contour-shift-range}.  The sector bound and conjugate
symmetry follow directly from the definition of \(\zeta\).

Finally,
\begin{align}
  \sum_k\frac{|\omega_k|}{|z_k|^r}
  &=
  \frac{\kappa_{n_{\rm c}}T^{r-1}}{2\pi}
  \sum_k
  e^{\Re\zeta(k\kappa_{n_{\rm c}})}
  \frac{
    |\zeta'(k\kappa_{n_{\rm c}})|
  }{
    |\zeta(k\kappa_{n_{\rm c}})|^r
  }.
  \label{eq:abstract-weight-sum-proof}
\end{align}
The summand is dominated by an integrable, double-exponentially decaying
function, uniformly for \(0\le r\le1\).  Its trapezoidal sum is therefore
bounded independently of \(n_{\rm c}\), proving
\eqref{eq:abstract-contour-general-weight-bound}.  The three estimates in
\eqref{eq:abstract-contour-special-weight-bounds} correspond to
\(r=0,\frac12,1\).
\end{proof}

\subsection{Selected linear combinations of inverses}

\begin{lemma}[Selected linear combination of inverses]
\label{lem:selected-inverse-lcu}
For each $k\in\mathcal I$, let $U_k$ be a
$\beta$-normalized block encoding of an invertible matrix $K_k$, and
assume
\[
  s_{\min}(K_k)\ge s_*>0
  \,\;\text{for all }k\in\mathcal I.
\]
Suppose that the selected oracle
\[
  U_{\rm sel}
  :=
  \sum_{k\in\mathcal I}
  |k\rangle\langle k|\otimes U_k
\]
is available.  Let
\[
  F:=\sum_{k\in\mathcal I}c_kK_k^{-1},
  \,\;
  C_c:=\sum_{k\in\mathcal I}|c_k|,
\]
and assume standard LCU state-preparation oracles realizing the
coefficients $c_k/C_c$.

Then $F$ has a block encoding with normalization
\begin{equation}
  \alpha_F
  =
  \bigO\!\left(\frac{C_c}{s_*}\right).
  \label{eq:selected-inverse-lcu-normalization}
\end{equation}
Each application of this block encoding requires
\begin{equation}
  \widetildeO\!\left(\frac{\beta}{s_*}\right)
  \label{eq:selected-inverse-lcu-cost}
\end{equation}
queries to $U_{\rm sel}$ and $U_{\rm sel}^\dagger$, apart from
coefficient-state preparation and precision logarithms.
\end{lemma}

\begin{proof}
This is Proposition~14 of
\cite{TakahiraOhashiSogabeUsuda2022Contour}, applied to the block-diagonal
selected matrix.  Its Hermitian-extension construction covers the
non-Hermitian blocks \(K_k\).
\end{proof}

\subsection{Proof of the single-level contour construction}
\label{app:single-contour-proof}

\begin{proof}[Proof of \cref{lem:single-contour-quadrature}]
Set
\[
  \mathcal R_h(z):=(L_h+zI)^{-1}.
\]
Since \(L_h=L_h^\dagger\succeq0\), the spectral theorem gives
\[
  \norm{\mathcal R_h(z)}
  =
  \max_{\lambda\in\operatorname{spec}(L_h)}
  \frac{1}{|\lambda+z|}.
\]
For every \(z\) satisfying
\(|\arg z|\le\pi-\delta\),
\cref{lem:scalar-sector-geometry} yields
\[
  |\lambda+z|
  \ge
  \sin(\delta/2)(\lambda+|z|)
  \ge
  \sin(\delta/2)|z|,
  \,\; \lambda\ge0.
\]
Consequently,
\begin{equation}
  \norm{\mathcal R_h(z)}
  \le
  C_\delta |z|^{-1},
  \label{eq:proof-single-resolvent-sector-bound}
\end{equation}
uniformly in \(h\).  Moreover, \(\mathcal R_h\) is holomorphic on the
sector \(\mathfrak S_\delta\), because all of its poles lie on the
excluded negative real axis.

The inverse-Laplace representation
\eqref{eq:single-inverse-laplace} identifies
\[
  \mathcal U(T)=E_h(T)
\]
in \cref{lem:operator-valued-contour}.  We may therefore apply that lemma
with
\[
  \nu=1,
  \,\;
  M=C_\delta.
\]
Identifying the weights in the abstract lemma with those in the main text,
\(w_k:=\omega_k\), the conjugate symmetry and shift geometry in
\eqref{eq:single-contour-symmetry} and
\eqref{eq:single-contour-shift-range} follow from
\eqref{eq:abstract-contour-shift-range}.  The error estimate
\eqref{eq:single-contour-quadrature-error} follows from
\eqref{eq:abstract-contour-error}; because \(\nu=1\), its prefactor is
\(T^{\nu-1}=1\).  Finally, the two bounds in
\eqref{eq:single-contour-weight-bounds} are respectively the cases
\(r=1\) and \(r=\tfrac12\) of
\eqref{eq:abstract-contour-general-weight-bound}.

To make the quadrature error at most a fixed multiple of \(\eps\), it is
enough to choose \(n_{\rm c}\) so that
\[
  \frac{n_{\rm c}}{\log n_{\rm c}}
  \gtrsim
  \log(1/\eps).
\]
This is ensured by
\[
  n_{\rm c}
  =
  \bigO\!\left(
    \log(1/\eps)\log\log(1/\eps)
  \right),
\]
which proves \eqref{eq:single-contour-node-count}.
\end{proof}

\begin{proof}[Proof of \cref{thm:single-contour-block-encoding}]
Let \(z_k,w_k\) be the contour nodes and weights from
\cref{lem:single-contour-quadrature}, and let
\[
  \rho_k:=\sqrt{z_k}
\]
denote the principal square root.  The sector separation in
\eqref{eq:single-contour-shift-range} ensures that this branch is
well-defined and satisfies \(\Re\rho_k>0\).  For each contour node define
\begin{equation}
  \mathcal K_k
  :=
  \mathcal K_{A_h}(z_k)
  =
  \begin{bmatrix}
    \rho_k I&A_h^\dagger\\
    A_h&-\rho_k I
  \end{bmatrix}.
  \label{eq:proof-nodewise-first-order-system}
\end{equation}
Let
\[
  \Pi_{\rm pr}:=\begin{bmatrix}I&0\end{bmatrix}
\]
denote the projection onto the first block.  By
\eqref{eq:app-complex-resolvent-factorization},
\begin{equation}
  (L_h+z_kI)^{-1}
  =
  \rho_k^{-1}
  \Pi_{\rm pr}\mathcal K_k^{-1}\Pi_{\rm pr}^\dagger.
  \label{eq:proof-factorized-resolvent}
\end{equation}

We first establish uniform block-encoding and inverse scales for the
family \(\{\mathcal K_k\}_k\).  A block encoding of
\(A_h/\alpha_A\), together with controlled preparation of the scalar
\(\rho_k\), gives a coherently selected block encoding of
\(\mathcal K_k\) with a common normalization
\begin{equation}
  \beta_{n_{\rm c}}
  =
  \bigO\!\left(
    \alpha_A+\max_k|\rho_k|
  \right).
  \label{eq:proof-selected-K-normalization}
\end{equation}
Indeed, the off-diagonal Hermitian dilation has normalization
\(\alpha_A\), while the diagonal part has normalization \(|\rho_k|\).
By the upper shift bound in
\eqref{eq:single-contour-shift-range},
\[
  \max_k|\rho_k|
  =
  \sqrt{\max_k|z_k|}
  =
  \bigO\!\left(
    \sqrt{\frac{n_{\rm c}}{T}}
  \right),
\]
and hence
\begin{equation}
  \beta_{n_{\rm c}}
  =
  \bigO\!\left(
    \alpha_A+\sqrt{\frac{n_{\rm c}}{T}}
  \right).
  \label{eq:proof-K-block-normalization}
\end{equation}

Next, \eqref{eq:app-complex-K-sector-bound} and the lower shift bound in
\eqref{eq:single-contour-shift-range} imply
\begin{align}
  s_{\min}(\mathcal K_k)
  &\ge
  c_{\delta_0}|\rho_k|
   =
  c_{\delta_0}\sqrt{|z_k|}
   \ge
  \frac{c}{\sqrt T}.
\end{align}
Therefore the selected family has the uniform inverse threshold
\begin{equation}
  s_*:=
  \min_k s_{\min}(\mathcal K_k)
  =
  \Omega(T^{-1/2}).
  \label{eq:proof-uniform-inverse-gap}
\end{equation}

Substituting \eqref{eq:proof-factorized-resolvent} into the contour sum
gives
\begin{align}
  \sum_{k=-n_{\rm c}}^{n_{\rm c}}
    w_k(L_h+z_kI)^{-1}
  &=
  \Pi_{\rm pr}
  \left(
    \sum_{k=-n_{\rm c}}^{n_{\rm c}}
      \gamma_k\mathcal K_k^{-1}
  \right)
  \Pi_{\rm pr}^\dagger,
  \label{eq:proof-factorized-contour-sum}\\
  \gamma_k
  &:=
  \frac{w_k}{\rho_k}.
  \notag
\end{align}
The coefficient estimate in
\eqref{eq:single-contour-weight-bounds} gives
\begin{equation}
  C_\gamma
  :=
  \sum_k|\gamma_k|
  =
  \sum_k
  \frac{|w_k|}{\sqrt{|z_k|}}
  =
  \bigO(T^{-1/2}).
  \label{eq:proof-gamma-mass}
\end{equation}

We can now apply \cref{lem:selected-inverse-lcu} to the selected family
\(\{\mathcal K_k\}_k\) with coefficients \(\gamma_k\).  Using
\eqref{eq:proof-uniform-inverse-gap} and
\eqref{eq:proof-gamma-mass}, the resulting block encoding of the operator
in parentheses in
\eqref{eq:proof-factorized-contour-sum} has normalization
\begin{equation}
  \alpha_{\rm cont}
  =
  \bigO\!\left(
    \frac{C_\gamma}{s_*}
  \right)
  =
  \bigO(1).
  \label{eq:proof-contour-block-normalization}
\end{equation}
Compression by $\Pi_{\rm pr}$ does not increase the normalization.
Thus the discrete contour sum has $\bigO(1)$ block-encoding
normalization, which is sufficient for
\cref{prop:direct-observable-from-E}. We retain this normalization
throughout and do not perform a separate amplification step.

The selected inverse cost from
\eqref{eq:selected-inverse-lcu-cost} is
\begin{align}
  \widetildeO\!\left(
    \frac{\beta_{n_{\rm c}}}{s_*}
  \right)
  &=
  \widetildeO\!\left[
    \left(
      \alpha_A+\sqrt{\frac{n_{\rm c}}{T}}
    \right)\sqrt T
  \right]
  \notag
  &=
  \widetildeO\!\left(
    1+\alpha_A\sqrt T+\sqrt{n_{\rm c}}
  \right).
  \label{eq:proof-inverse-condition-scale}
\end{align}
Each query to the selected block encoding of \(\mathcal K_k\) uses only
a constant number of queries to the block encoding of \(A_h\) and its
adjoint; the dependence on \(k\) is handled coherently by reversible
arithmetic for \(z_k\) and \(\rho_k\).

It remains to combine the approximation errors.  Choose
\(n_{\rm c}\) according to
\eqref{eq:single-contour-node-count}, so that the contour error in
\eqref{eq:single-contour-quadrature-error} is a fixed fraction of
\(\eps\).  Approximate the selected inverses, the coefficient-state
preparation, and the reversible arithmetic to sufficiently small fixed
fractions of the remaining error budget.  Restoring the logarithmic
precision dependence hidden in
\eqref{eq:selected-inverse-lcu-cost} gives
\[
  \mathcal T_{\rm cont}(h,T,\eps)
  =
  \bigO\!\left(
    \left[
      1+\alpha_A\sqrt T+\sqrt{n_{\rm c}}
    \right]
    \log\!\frac{
      n_{\rm c}(1+\alpha_A\sqrt T)
    }{\eps}
  \right),
\]
which is \eqref{eq:single-contour-query-complexity}.  Since
\(n_{\rm c}\) is polylogarithmic in \(1/\eps\),
\[
  \mathcal T_{\rm cont}(h,T,\eps)
  =
  \widetildeO\!\left(
    1+\alpha_A\sqrt T
  \right).
\]
Finally, for discrete-gradient access,
\(\alpha_A=\bigO(h^{-1})\).  In the resolved positive-time regime
\(T\ge1\) and \(h\le1\), the unit term is dominated, and therefore
\[
  \mathcal T_{\rm cont}(h,T,\eps)
  =
  \widetildeO\!\left(
    \frac{\sqrt T}{h}
  \right),
\]
which proves \eqref{eq:single-contour-cost-h}.
\end{proof}

\section{Estimates for the two-level transfer map}
\label{app:joint-transfer-contour-proofs}

\subsection{Proof of positive-time interlevel smoothing}
\label{app:interlevel-smoothing}

\begin{proof}[Proof of \cref{lem:interlevel-smoothing}]
Under the mass-reduced isometries between \(\C^{N_j}\) and \(V_j\),
the operators \(e^{-tL_j}\) represent the Galerkin semigroups
\(\mathcal E_j(t)\), while \(P_\ell\) represents the inclusion
\(V_{\ell-1}\hookrightarrow V_\ell\) and \(P_\ell^\dagger\) the
corresponding \(L^2\)-projection.  Comparing both Galerkin semigroups
with the common continuous semigroup \(\mathcal E(t)=e^{-t\mathcal L}\),
using the standard nonsmooth-data estimate
\[
  \norm{\mathcal E(t)-\mathcal E_j(t)\Pi_j}_{L^2\to L^2}
  \le
  C\min\left\{1,\frac{h_j^2}{t}\right\},
\]
and applying the triangle inequality together with
\(h_{\ell-1}=2h_\ell\), gives
\[
  \norm{
    e^{-tL_\ell}
    -
    P_\ell e^{-tL_{\ell-1}}P_\ell^\dagger
  }
  \le
  C\min\left\{1,\frac{h_\ell^2}{t}\right\}.
\]
Applying this estimate to the compatible initial data proves
\eqref{eq:interlevel-final-time}.
\end{proof}

\subsection{The corrected resolvent as a transfer function}

\begin{proof}[Proof of \cref{lem:joint-transfer-resolvent}]
We first observe that 
\begin{equation}
 \bm\phi_\ell(0)=
  \begin{bmatrix}
    h_\ell^{-2}(I-P_\ell P_\ell^\dagger)\\[1mm]
    P_\ell^\dagger
  \end{bmatrix}\by_\ell(0).
  \label{eq:app-compatible-data-embedding}
\end{equation}
The block inverse of \(zI+\widehat{L}_\ell\) is
\begin{equation}
  (zI+\widehat{L}_\ell)^{-1}
  =\begin{bmatrix}
    R_\ell(z)&-h_\ell^{-2}R_\ell(z)B_\ell R_{\ell-1}(z)\\[1mm]
    0&R_{\ell-1}(z)
  \end{bmatrix}.
  \label{eq:app-joint-block-inverse}
\end{equation}
The shifted defect identity
\begin{equation}
  (L_\ell+zI)P_\ell-P_\ell(L_{\ell-1}+zI)=B_\ell
  \label{eq:app-shifted-defect-identity}
\end{equation}
implies
\begin{equation}
  R_\ell(z)B_\ell R_{\ell-1}(z)
  =P_\ell R_{\ell-1}(z)-R_\ell(z)P_\ell.
  \label{eq:app-shifted-defect-cancellation}
\end{equation}
Multiplying \eqref{eq:app-joint-block-inverse} by
\eqref{eq:app-compatible-data-embedding} gives
\begin{equation}
  (zI+\widehat{L}_\ell)^{-1}\begin{bmatrix}
    h_\ell^{-2}(I-P_\ell P_\ell^\dagger)\\[1mm]
    P_\ell^\dagger
  \end{bmatrix}
  =\begin{bmatrix}
    h_\ell^{-2}D_\ell(z)\\[1mm]
    R_{\ell-1}(z)P_\ell^\dagger
  \end{bmatrix}.
  \label{eq:app-exact-transfer-resolvent}
\end{equation}
Taking the first row gives \eqref{eq:correction-transfer-resolvent-z}.
Sectorial deformation of the Bromwich line then proves
\eqref{eq:correction-transfer-contour-z}.
\end{proof}

\subsection{Sectorial corrected-resolvent estimate}

\begin{lemma}[Sectorial two-level shifted-resolvent estimate]
\label{lem:sectorial-corrected-resolvent}
Under the assumptions of \cref{lem:interlevel-smoothing}, for every fixed
\(\delta\in(0,\pi)\),
\begin{equation}
  \norm{D_\ell(z)}
  \le
  C_\delta
  \frac{h_\ell^2}{1+|z|h_\ell^2},
  \,\;
  |\arg z|\le\pi-\delta.
  \label{eq:app-corrected-resolvent-bound}
\end{equation}
Furthermore, if \(L_j\succeq\lambda_*I\) uniformly in \(j\), then
\begin{equation}
  \norm{R_j(z)}
  \le
  \frac{C_\delta}{\lambda_*+|z|}.
  \label{eq:app-coarse-resolvent-bound}
\end{equation}
\end{lemma}

The proof uses the standard sectorial Galerkin resolvent approximation:
the complex-shifted variational problem remains uniformly stable after a
sector-dependent rotation, and the usual C\'ea--Aubin--Nitsche argument
gives an \(L^2\)-error of order \(h^2\).  Closely related
parameter-dependent resolvent estimates are used in contour-based finite
element approximations of parabolic and fractional elliptic evolution
problems; see, for example,
\cite{SheenSloanThomee1999,SheenSloanThomee2003,
McLeanThomee2011}.
The two-level estimate below follows by applying the one-level estimate on
the two adjacent spaces and comparing both with the same continuous
resolvent.

\begin{proof}
The second estimate follows immediately from the spectral theorem and
\cref{lem:scalar-sector-geometry}.  Indeed, for
\(\lambda\ge\lambda_*\),
\[
  |\lambda+z|
  \ge
  c_\delta(\lambda+|z|)
  \ge
  c_\delta(\lambda_*+|z|),
\]
and hence
\[
  \norm{R_j(z)}
  =
  \max_{\lambda\in\operatorname{spec}(L_j)}
  \frac1{|\lambda+z|}
  \le
  \frac{C_\delta}{\lambda_*+|z|}.
\]
Without a uniform gap, the same argument gives
\(\norm{R_j(z)}\le C_\delta|z|^{-1}\).

We prove the corrected-resolvent estimate in two regimes.  First suppose
that
\[
  |z|h_\ell^2\ge1.
\]
Since \(P_\ell\) is an isometry,
\[
\begin{aligned}
  \norm{D_\ell(z)}
  &\le
  \norm{R_\ell(z)}
  +
  \norm{P_\ell R_{\ell-1}(z)P_\ell^\dagger}  \\
  &=
  \norm{R_\ell(z)}
  +
  \norm{R_{\ell-1}(z)}
  \le
  \frac{C_\delta}{|z|}.
\end{aligned}
\]
Because \(|z|h_\ell^2\ge1\),
\[
  \frac1{|z|}
  \le
  2\frac{h_\ell^2}{1+|z|h_\ell^2},
\]
which proves the desired estimate in this regime.

It remains to consider
\[
  |z|h_\ell^2\le1.
\]
Let \(\Pi_j:L^2(\Omega)\to V_j\) be the \(L^2\)-orthogonal
projection.  Define the continuous resolvent
\[
  \mathcal R(z):=(\mathcal L+zI)^{-1},
\]
and the discrete resolvent
\(\mathcal R_j(z):V_j\to V_j\) by
\[
  \mathrm a_z(\mathcal R_j(z)g_j,v_j)
  =
  (g_j,v_j)_{L^2},
  \,\; g_j,v_j\in V_j,
\]
where, on the complexification of \(H_0^1(\Omega)\),
\[
  \mathrm a_z(u,v)
  :=
  \mathrm a(u,v)+z(u,v)_{L^2}.
\]
Thus the Galerkin approximation to \(\mathcal R(z)f\) is
\(\mathcal R_j(z)\Pi_jf\).

Here we distinguish between resolvents in coefficient and function space.
The matrix resolvent
\(
R_j(z)=(L_j+zI)^{-1}
\)
acts on mass-reduced coordinate vectors in \(\mathbb C^{N_j}\), whereas
the calligraphic resolvent \(\mathcal R_j(z)\) below acts on finite-element
functions in \(V_j\).  The two are unitarily equivalent under the
mass-reduced isometry \(\mathcal U_j:\mathbb C^{N_j}\to V_j\).

Writing \(z=|z|e^{\ii\theta}\), we have
\[
  \Re\!\left[
    e^{-\ii\theta/2}\mathrm a_z(v,v)
  \right]
  =
  \cos(\theta/2)
  \left(
    \mathrm a(v,v)+|z|\norm{v}_{L^2}^2
  \right).
\]
Since \(|\theta|\le\pi-\delta\), the shifted form is uniformly
coercive and continuous in the norm
\[
  \norm{v}_z^2
  :=
  \mathrm a(v,v)+|z|\norm{v}_{L^2}^2.
\]
The sectorial resolvent bound and elliptic \(H^2\)-regularity apply
uniformly to both the primal problem and its adjoint.  Hence the
standard C\'ea and Aubin--Nitsche arguments give, for every fixed
\(C_0>0\),
\begin{equation}
  \norm{
    \mathcal R(z)-\mathcal R_j(z)\Pi_j
  }_{L^2\to L^2}
  \le
  C_{\delta,C_0}h_j^2,
  \,\;
  |z|h_j^2\le C_0.
  \label{eq:app-shifted-galerkin-resolvent-error}
\end{equation}
Indeed, the corresponding energy- and duality-error estimates are
\[
  \norm{u-u_j}_z
  \le
  C_{\delta,C_0}h_j\norm{f}_{L^2},
  \,\;
  \norm{u-u_j}_{L^2}
  \le
  C_{\delta,C_0}h_j\norm{u-u_j}_z.
\]

Now to prove \eqref{eq:app-corrected-resolvent-bound}, we identify the matrix \(D_\ell(z)\) with the difference of the
two Galerkin resolvents.  Let
\(\mathcal U_j:\C^{N_j}\to V_j\) denote the mass-reduced isometry.
By construction,
\begin{equation}
  \mathcal U_jR_j(z)\mathcal U_j^\dagger
  =
  \mathcal R_j(z)\Pi_j
  \label{eq:app-discrete-resolvent-identification}
\end{equation}
as an operator on \(L^2(\Omega)\), with the output naturally embedded
in \(L^2(\Omega)\).  Moreover, the mass-scaled prolongation represents
the natural finite-element inclusion, and therefore
\[
  \mathcal U_\ell P_\ell=\mathcal U_{\ell-1},
  \,\;
  P_\ell^\dagger\mathcal U_\ell^\dagger
  =
  \mathcal U_{\ell-1}^\dagger,
\]
where both sides are regarded as maps to or from \(L^2(\Omega)\).
Consequently,
\begin{align}
  \mathcal U_\ell D_\ell(z)\mathcal U_\ell^\dagger
  &=
  \mathcal U_\ell R_\ell(z)\mathcal U_\ell^\dagger
  -
  \mathcal U_\ell P_\ell
  R_{\ell-1}(z)P_\ell^\dagger
  \mathcal U_\ell^\dagger
  \notag\\
  &=
  \mathcal R_\ell(z)\Pi_\ell
  -
  \mathcal R_{\ell-1}(z)\Pi_{\ell-1}.
  \label{eq:app-D-continuous-identification}
\end{align}
Because \(\mathcal U_\ell\) is an isometry,
\[
  \norm{D_\ell(z)}
  =
  \norm{
    \mathcal U_\ell D_\ell(z)\mathcal U_\ell^\dagger
  }_{L^2\to L^2}.
\]

Since \(h_{\ell-1}=2h_\ell\), the condition
\(|z|h_\ell^2\le1\) implies
\(|z|h_{\ell-1}^2\le4\).  Thus
\eqref{eq:app-shifted-galerkin-resolvent-error} applies on both
levels, and comparison with the same continuous resolvent yields, by triangle inequality,
\[
\begin{aligned}
  \norm{D_\ell(z)}
  &\le
  \norm{
    \mathcal R_\ell(z)\Pi_\ell-\mathcal R(z)
  }_{L^2\to L^2}
  +
  \norm{
    \mathcal R(z)-\mathcal R_{\ell-1}(z)\Pi_{\ell-1}
  }_{L^2\to L^2} \\
  &\le
  C_\delta
  \left(h_\ell^2+h_{\ell-1}^2\right)
  \le
  C_\delta h_\ell^2.
\end{aligned}
\]
Since \(1+|z|h_\ell^2\le2\) in this regime,
\[
  h_\ell^2
  \le
  2\frac{h_\ell^2}{1+|z|h_\ell^2}.
\]
Together with the large-shift estimate, this proves
\eqref{eq:app-corrected-resolvent-bound}.
\end{proof}

\subsection{Proof of the contour-transfer results}\label{proof-contour-correction}

\begin{proof}[Proof of \cref{thm:scaled-correction-contour}]
By the assumed estimate \eqref{eq:sectorial-two-level-estimates},
\begin{equation}
  \norm{h_\ell^{-2}D_\ell(z)}
  \le \frac{C}{1+|z|h_\ell^2}
  \le C
  \label{eq:app-correction-symbol-bound}
\end{equation}
on the admissible sector.  Apply \cref{lem:operator-valued-contour} with
\(\nu=0\) to the transfer function
\(h_\ell^{-2}D_\ell(z)\).  This gives the node geometry and the error
\eqref{eq:correction-transfer-quadrature-error}.  The same contour estimate,
or directly the general weight estimate in
\eqref{eq:abstract-contour-general-weight-bound}, gives
\eqref{eq:correction-transfer-weight-bound}.

For the LCU normalization,
\[
  \sum_k|\omega_k|h_\ell^{-2}
  \alpha_{D_\ell(z_k)}
  \le
  C_D\sum_k|\omega_k|
  =
  \bigO(T^{-1}).
\]
PREPARE--SELECT--UNPREPARE over the contour register therefore block encodes
\(\widehat D_\ell(T)\) at the stated normalization.  The coherently
selected implementation has the largest branch cost, under the stated
uniform selected-access assumption. If its nodewise approximants satisfy
$\|\widetilde D_\ell(z_k)-D_\ell(z_k)\|\le\eta h_\ell^2$, then
\[
 \sum_k |\omega_k|h_\ell^{-2}
 \|\widetilde D_\ell(z_k)-D_\ell(z_k)\|
 \le C\eta/T.
\]
Thus normalized branch accuracy $O(\eta)$ suffices, with $\eta$
chosen from the physical observable error budget. Its cost is
polylogarithmic in $\eta^{-1}$.  This proves
\eqref{eq:correction-transfer-block-normalization}.
\end{proof}

\end{document}